\documentclass[sigconf]{acmart}
\usepackage{bm}
\usepackage{pifont}
\usepackage{algorithm}
\usepackage{subcaption}
\usepackage{multirow}
\usepackage{algorithmic}
\newcommand\numberthis{\addtocounter{equation}{1}\tag{\theequation}}
\newtheorem*{remark}{Remark}
\AtBeginDocument{%
  }

\copyrightyear{2026}
\acmYear{2026}
\setcopyright{cc}
\setcctype{by}
\acmConference[KDD '26]{Proceedings of the 32nd ACM SIGKDD Conference on Knowledge Discovery and Data Mining V.2}{August 09--13, 2026}{Jeju Island, Republic of Korea}
\acmBooktitle{Proceedings of the 32nd ACM SIGKDD Conference on Knowledge Discovery and Data Mining V.2 (KDD '26), August 09--13, 2026, Jeju Island, Republic of Korea}
\acmDOI{10.1145/3770855.3817611}
\acmISBN{979-8-4007-2259-2/2026/08}

\begin{document}

\title{Enforcing Trust Accountability with Backward Propagation}

\author{Wenbo Wu}
\authornote{The corresponding author.}
\orcid{0009-0002-3937-0124}
\affiliation{%
  \institution{University of Southampton}
  \city{Southampton}
  \country{UK}
}
\email{wenbo.wu@soton.ac.uk}

\author{George Konstantinidis}
\orcid{0000-0002-3962-9303}
\affiliation{%
  \institution{University of Southampton}
  \city{Southampton}
  \country{UK}
}
\email{g.konstantinidis@soton.ac.uk}

\begin{abstract}
Trust and reputation management underpins reliable interactions in distributed networks, yet existing trust models rely solely on forward propagation of interaction-based trust signals. They lack robust mechanisms to enforce accountability for the propagated trust signals when negative interactions occur.
In addition, such models often fail to initialize newly joined nodes with sparse interaction history, leading to the cold-start problem.
In this paper, we propose \emph{RepuLink}, a two-layer reputation model that couples an endorsement network with an interaction feedback network. 
RepuLink integrates two concurrent \emph{backward propagation} mechanisms: 
Backward Endorsement Penalty Propagation~(BEPP), which recursively penalizes endorsers of misbehaving nodes, and Backward Endorsement Reward Propagation~(BERP), which rewards endorsers of well-performing nodes. Together, RepuLink enforces endorsement accountability and incentivizes positive behaviors, which form a positive interaction feedback loop.
The endorsement layer further provides explainable, endorser-weighted trust initialization for newly joined nodes. 
Experiments on real-world datasets against representative trust propagation baselines demonstrate that RepuLink outperforms across four evaluation metrics in both interaction-only and full two-layer settings, while preserving comparable efficiency.
\end{abstract}



\maketitle
\newcommand\kddavailabilityurl{https://doi.org/10.5281/zenodo.20289640}
\ifdefempty{\kddavailabilityurl}{}{
\begingroup\small\noindent\raggedright\textbf{Resource Availability:}\\
The source code of this paper has been made publicly available at \url{\kddavailabilityurl}.
\endgroup
}

\section{Introduction}
\label{sec:introduction}
Trust is the foundation of reliable interactions in distributed networks~\cite{wu2025trust}, and is crucial for encouraging active user engagement in applications such as data markets~\cite{ma2024model,fernandez2020data,nguyen2025blockchain} and autonomous vehicles~\cite{he2022security,tan2023digital}.
Trust and reputation management systems~\cite{granatyr2015trust,su2015reliable} have emerged as effective methods to quantify the trustworthiness of participants based on their behavior or interaction history in large, anonymous, distributed settings.

Commonly, a trustworthiness score is first assigned to nodes based on individual evaluations or interactions, and then aggregated to calculate node reputations. 
Simple trust aggregation systems, such as those discussed in~\cite{josang2007survey} or those found in bilateral rating platforms (\textit{e.g.}, Uber and Airbnb), typically rely on averaging received scores for each node. While effective at identifying peers with very high or very low reputations, such approaches lack the ability to capture transitive trust relationships within the network.

More advanced algorithms address this limitation by \emph{propagating} trust signals through the network~\cite{jiang2016understanding}: the evaluation of an entity $B$ by another entity $A$ is weighted by $A$'s own reputation, applied recursively. 
PageRank~\cite{page1999pagerank} is a core algorithm to perform this task, and other approaches include EigenTrust~\cite{kamvar2003eigentrust} and PowerTrust~\cite{zhou2007powertrust}. 
The most recent trust models in this space are AbsoluteTrust~\cite{awasthi2020absolutetrust} and ShapleyTrust~\cite{bandhana2024trust}, which explore meaningful trust evaluation and propagation methods.

\begin{figure}[t]
    \centering
    \includegraphics[width=0.8\linewidth]{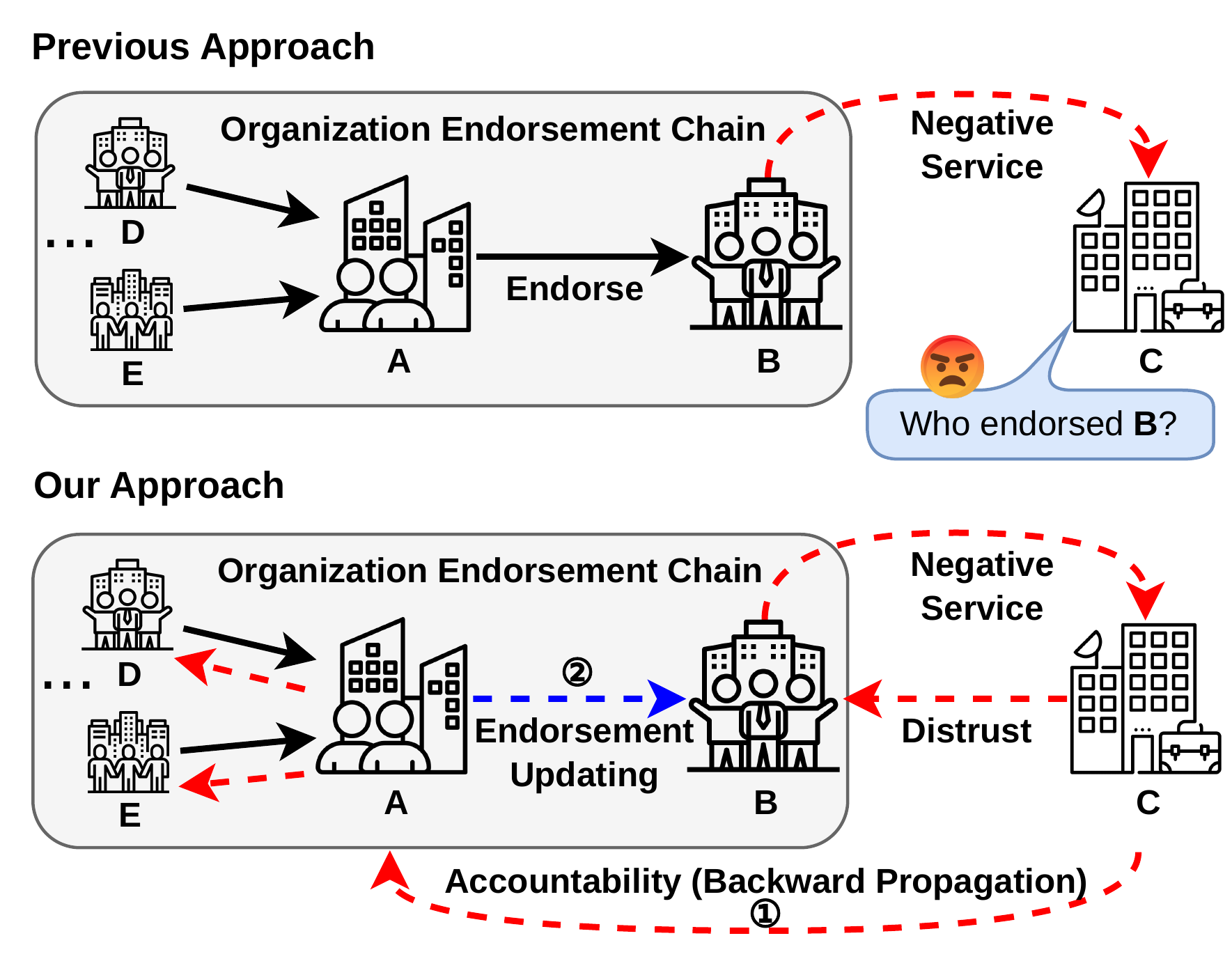}
    \caption{Traditional Trust vs. Accountable Trust}
    \label{fig:example}
\end{figure}

Despite these advances, existing trust models exhibit three key limitations.
First, they primarily rely on interaction-based node evaluation and often lack support for domain-specific knowledge that can capture known trust relationships (\textit{e.g.}, social information, institutional endorsement within an organization or a consortium).
Second, these systems face a well-known \emph{cold-start problem}~\cite{feng2021rbpr,bampatsikos2021solving}: they cannot reliably evaluate newly joined or sparsely connected nodes that lack sufficient interaction history. Leveraging explicit endorsement knowledge can address this gap by assigning explainable initial reputation scores proportional to the endorsers' own reputations.
Third, existing trust models rely exclusively on \emph{forward} propagation of (positive) trust signals and lack mechanisms for \emph{accountability} of the propagated trust signals when negative interactions occur.

Consider a scenario as in Fig.~\ref{fig:example}: organization~A endorses organization~B based on institutional knowledge or domain-specific information; subsequently, organization~C had a negative service experience with~B.
This conflict makes it essential to trace accountability back through~B's endorsement chain, since B's endorsers should be held accountable for their endorsements. The absence of such a backward accountability mechanism can lead to skewed reputation scores and render systems vulnerable to malicious collusive behaviors, where groups of nodes strategically post positive feedback to each other to inflate reputations without consequence.

\begin{table}[t]
    \centering
    \caption{Comparison of existing trust/reputation models.}
    \label{tab:role-specific-pagerank}
    \resizebox{\columnwidth}{!}{%
    \begin{tabular}{lcccc}
        \toprule
        \textbf{Algorithm} & 
        \begin{tabular}[c]{@{}c@{}}\textbf{Forward} \\ \textbf{Propagation}\end{tabular} & 
        \begin{tabular}[c]{@{}c@{}}\textbf{Incentive} \\ \textbf{Mechanism}\end{tabular} & 
        \begin{tabular}[c]{@{}c@{}}\textbf{Accountability} \\ \textbf{Mechanism}\end{tabular}& 
        \begin{tabular}[c]{@{}c@{}}\textbf{Meaningful} \\ \textbf{Initialization}\end{tabular}\\
        \midrule
        PageRank \cite{page1999pagerank} & \ding{51} & \ding{55} & \ding{55} & \ding{55} \\
        EigenTrust \cite{kamvar2003eigentrust} & \ding{51} & \ding{55} & \ding{55} & \ding{55} \\
        PowerTrust \cite{zhou2007powertrust} & \ding{51} & \ding{55} & \ding{55} & \ding{55} \\
        AbsoluteTrust \cite{awasthi2020absolutetrust} & \ding{51} & \ding{55} & \ding{55} & \ding{55} \\
        ShapleyTrust  \cite{bandhana2024trust} & \ding{51} & \ding{55} & \ding{55} & \ding{55} \\
        \textbf{RepuLink}   & \ding{51} & \ding{51} & \ding{51} & \ding{51}  \\
        \bottomrule
    \end{tabular}}
\end{table}

In this paper, we introduce \emph{RepuLink}, a novel reputation model that addresses previous limitations (see Table~\ref{tab:role-specific-pagerank}). 
RepuLink is a two-layered model that combines a domain-specific endorsement network with an interaction feedback network. 
A key novelty lies in the two concurrent backward propagation mechanisms:
(1) a penalizing/accountability mechanism, \emph{Backward Endorsement Penalty Propagation}~(BEPP), which recursively penalizes endorsers whose endorsees receive negative interaction feedback; and
(2) a rewarding/incentive mechanism, \emph{Backward Endorsement Reward Propagation}~(BERP), which recursively rewards endorsers whose endorsees demonstrate consistently positive performance.

Our contributions are as follows:
\begin{itemize}
    \item A formal, two-layered reputation model that integrates endorsements with interaction-based feedback, enabling richer trust assessment (Section~\ref{sec:model-design}).
    
    \item An effective solution to the cold-start problem that leverages the endorsement layer to assign explainable initial reputation values to newly joined nodes, proportional to the collective reputation of their direct and indirect endorsers (Section~\ref{subsec:node arrival and departure}).

    \item Two novel backward propagation mechanisms (BEPP and BERP) that dynamically and recursively propagate penalties and rewards through the endorsement chain, providing practical accountability and incentives. An ablation study confirms that BERP drives large-scale rank reordering while BEPP acts as a targeted corrective mechanism (Section~\ref{subsec:impact of interactions}).

    \item An extensive experimental evaluation on real-world datasets against five other trust models. RepuLink consistently outperforms all baselines in both interaction-only and full two-layer settings across four major evaluation metrics, with comparable efficiency (Section ~\ref{sec:experiments}).
    
    \item A formal analysis of RepuLink, including proofs of convergence for both forward and backward propagation, convergence speed analysis, and computational complexity, is provided in supplementary materials (Appendix~\ref{sec:proof-of-convergence},~\ref{sec:convergence-speed-analysis},~\ref{sec:computational-complexity-analysis}).
\end{itemize}

\begin{figure*}
    \centering
    \includegraphics[width=0.8\linewidth]{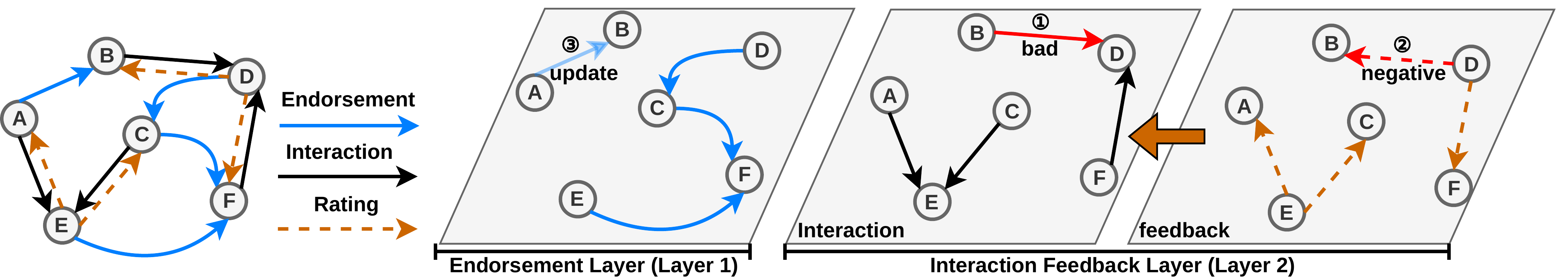}
    \caption{RepuLink Model.
    B performed badly in Layer~2 and received negative feedback from~D.
    A as the endorser of B will be penalized through the accountability mechanism.
    The endorsement confidence from A toward~B will be updated in Layer~1.}
    \label{fig:repulink}
\end{figure*}

\section{The RepuLink Model}
\label{sec:model-design}
This section presents the formal foundations of RepuLink, including the computation of reputation and the backward propagation of endorsement accountability.
A \emph{trust and reputation network} is a two-layer directed graph $\mathcal{G} = (\mathcal{V}, \mathcal{E}, \mathcal{F})$, where $\mathcal{V}$ is a set of $N$ nodes, and $\mathcal{E}$ and $\mathcal{F}$ are two layers of directed edges.
The first layer, $\mathcal{E} \subseteq \mathcal{V} \times \mathcal{V}$, consists of \emph{endorsement} relationships; intuitively, an edge $(i,j)$ is in $\mathcal{E}$ if node $i$ endorses the integrity of node $j$.
An \emph{endorsement confidence} function $\epsilon:\mathcal{E} \rightarrow [0,1]$ assigns to each edge $(i,j) \in \mathcal{E}$ a weight, $\epsilon(i,j)$, that quantifies the endorsement strength from node $i$ to node $j$.
The second layer, $\mathcal{F} \subseteq \mathcal{V} \times \mathcal{V}$, captures the \emph{interaction feedback} between nodes. An interaction feedback mapping function, $f:\mathcal{F} \rightarrow \mathbb{R}_{\geq 0} \times \mathbb{R}_{\geq 0}$, assigns to each edge $(i,j) \in \mathcal{F}$ a pair of scores, $f(i,j) = (p_{ij}, n_{ij})$, where $p_{ij}$/$n_{ij}$ represents the cumulative positive/negative interaction feedback that node $i$ gives to node $j$.
The \emph{reputation value} of a node is constructed by integrating knowledge from both the endorsement $\mathcal{E}$ and interaction feedback $\mathcal{F}$ networks, as in Fig.~\ref{fig:repulink}.

\subsection{Endorsement Network}
Given an endorsement layer $\mathcal{E}$ and its endorsement confidence function $\epsilon$, the endorsement strengths are represented by a matrix $\hat{E} \in \mathbb{R}^{N \times N}$, where an element $\hat{E}_{ij} = \epsilon(i,j)$ if $(i,j) \in \mathcal{E}$, and $\hat{E}_{ij} = 0$ otherwise.
A value of $\hat{E}_{ij}=0$ indicates no endorsement from node \(i\) (endorser) to node \(j\) (endorsee) (i.e., $(i,j) \notin \mathcal{E}$), while $\hat{E}_{ij}=1$ indicates a strong endorsement.
We normalize the endorsement as:
\begin{equation} \label{eq:normalised endorsement}
E_{ij} = \frac{\hat{E}_{ij}}{\sum_{k} \hat{E}_{ik} + c},
\end{equation}
where \(E \in \mathbb{R}^{N \times N}\) and \(\Vert E^{\top} \Vert_1 = 1\). Here, \(c\) is a small positive constant (\(0 < c \ll 1\)) introduced to ensure numerical stability and prevent division-by-zero errors.

\subsection{Interaction Feedback Network}
The interaction feedback layer $\mathcal{F}$, as defined, records cumulative feedback scores that are aggregated over potentially many individual interactions.
By aggregating the interaction feedback from node \(i\) (reviewer) to node \(j\) (reviewee), we get $f(i,j) = (p_{ij},n_{ij})$.
The \textit{trustworthiness} (which is also known as \textit{local trust} in \cite{kamvar2003eigentrust,awasthi2020absolutetrust}) of node \(j\) from the perspective of node \(i\), denoted $\hat{T}_{ij}$, is calculated as:
\begin{equation}
    \hat{T}_{ij} = \frac{p_{ij} - n_{ij}}{p_{ij} + n_{ij} + c}.
    \label{eq:trustworthiness}
\end{equation}
Note that, consistent with prior works such as~\cite{kamvar2003eigentrust}, this model imposes a significant penalty for negative interactions. Specifically, the negative feedback score $n_{ij}$ not only decreases the numerator of Eq.~\eqref{eq:trustworthiness}, but is also added to the denominator.
Before propagating trustworthiness, each node performs a local normalization to prevent arbitrary assignments of untruthful trustworthiness values. We define the \textit{normalized trustworthiness} as:
\begin{equation}
    T_{ij} = \frac{\max(\hat{T}_{ij},0)}{\sum_k \max(\hat{T}_{ik},0) + c},
    \label{eq:normalized_trustworthiness}
\end{equation}
where \(T \in \mathbb{R}^{N \times N}\) and \(\Vert T^{\top} \Vert_1 = 1\).

\subsection{Reputation Function}
\label{subsec:reputation function}
In our model, the \textit{reputation} of a network node (similar to what is known as \textit{global trust} in \cite{kamvar2003eigentrust,awasthi2020absolutetrust}) is maintained by integrating knowledge from both the endorsement and the interaction layers. This integration is achieved via a weighted sum, with a parameter $\alpha \in [0,1]$ that balances the contributions: $\alpha$ weights the knowledge from interactions, while $(1-\alpha)$ weights that from endorsements.

The vector of node reputations at time $t$ is denoted by $R^{(t)} = (R_1^{(t)}, R_2^{(t)}, \dots, R_N^{(t)})^\top \in \mathbb{R}^N$.
At the initial stage (\(t = 0\)), there is an absence of interaction history. This is the \textit{cold-start} issue, commonly encountered by most trust and reputation management systems~\cite{tahta2015gentrust,meng2020truetrust}.
To alleviate this issue, we initialize the reputation of node $j$, taking only endorsements into account in a uniform way: $R_j^{(0)} = (1-\alpha) \sum_{k \in [1,N]} E_{kj} \frac{1}{N}$.
The reputation of node \(j\) at $t+1$ is:
\begin{small}
\begin{align*}
R_j^{(t+1)}
&= \alpha \sum_{k \in [1,N]} T_{kj} R_k^{(t)}
+ (1 - \alpha) \sum_{k \in [1,N]} E_{kj} R_k^{(t)} \\[8pt]
&= \alpha [T^\top R^{(t)}]_j + (1 - \alpha) [E^\top R^{(t)}]_j \\[8pt]
&= \left[ (\alpha T^\top + (1 - \alpha) E^\top) R^{(t)} \right]_j
= \left[W R^{(t)}\right]_j. \numberthis
\end{align*}
\end{small}
The vector form is succinctly expressed as:
\begin{small}
\begin{equation}
    R^{(t+1)} = W R^{(t)}, \quad \text{where} \quad W = \alpha T^\top + (1 - \alpha)E^\top.
    \label{eq:reputation}
\end{equation}
\end{small}
At any timepoint $t$, to ensure comparability and numerical stability, we project the reputation vector onto the non-negative \( \ell_1 \) simplex. That is, we always override the latest reputation vector by:
\begin{small}
\begin{equation}
    R^{(t)} = \frac{\max\left( R^{(t)}, 0 \right)}{\left\| \max\left( R^{(t)}, 0 \right) \right\|_1 + c}.
    \label{eq:reputation-normalization}
\end{equation}
\end{small}

In dynamic network environments, node behavior is inherently transient: nodes join and depart, endorsements evolve, and interactions continuously generate new feedback. Consequently, the reputation calculation via Eq.~\eqref{eq:reputation} must adapt, consistently utilizing the latest effective endorsement matrix \( E \) and the normalized trustworthiness matrix \( T \) to reflect current network conditions.
We detail the processes for handling node arrivals and departures in Sec.~\ref{subsec:node arrival and departure}. Subsequently, we focus on internode interactions, examining the distinct impacts of positive and negative feedback in Sec.~\ref{subsec:impact of interactions}.

\subsection{Node Arrival and Departure}
\label{subsec:node arrival and departure}
\noindent
\textbf{Arrivals of New Nodes.}
At any time slot \(t\), new nodes may freely join the network. To assign a fair and meaningful initial reputation for these newly joined nodes, we extend the endorsement network to take into account the domain-specific knowledge about node relationships and overcome the cold-start issue.
RepuLink assigns a meaningful initial reputation to new nodes that is directly proportional to the collective reputation of their endorsers.
This design ensures that nodes endorsed by high-reputation nodes receive a high initial score, while those backed by low-reputation nodes are assigned a correspondingly low score. We empirically validate this contextual assignment mechanism in Sec.~\ref{subsec:effectiveness evaluation}.
Formally, given a trust and reputation network $\mathcal{G} = (\mathcal{V}, \mathcal{E}, \mathcal{F})$, a newly joined node $j \notin \mathcal{V}$, a set of existing nodes $S_j \subseteq \mathcal{V}$ that endorse $j$, and a confidence function $\epsilon':S_j \rightarrow [0,1]$ that captures these endorsements, our trust and reputation network is updated as:

\textit{Step 1: Endorsement Set Construction.}
Add the new node $j$ to $\mathcal{V}$ (i.e., $\mathcal{V} := \mathcal{V} \cup \{j\}$). The endorsement layer $\mathcal{E}$ and confidence function $\epsilon$ are updated to include the new endorsements for $j$:
$\mathcal{E} := \mathcal{E} \cup \{ (i,j) \mid i \in S_j \}$;
$\epsilon := \epsilon \cup \epsilon'$.

\textit{Step 2: Initial Reputation Calculation.}
To assign an initial reputation for node $j$ based on endorsements, we utilize the endorsers' reputations at the current time slot $t$:
$R_j^{(t)} = (1-\alpha) \sum_{i \in S_j} E_{ij} R_i^{(t)}$,
where \( E_{ij} \) is the most current normalized endorsement from node \( i \) to node \( j \), and \( R_i^{(t)} \) is the reputation of endorser \( i \) at the current time slot $t$.
The overall reputation vector $R^{(t)}$, now including $R_j^{(t)}$, is then normalized via Eq.~\eqref{eq:reputation-normalization}.

\noindent
\textbf{Departure of Nodes.}
In decentralized networks, node departure is a common phenomenon. Our model handles this by preserving essential information in case of nodes rejoining the network.
When a node leaves the network, its identity is not removed but marked as \textit{deactivated}.
While a node is deactivated, it is excluded from consideration in the reputation computations of active network participants.
All historical interaction feedback data associated with the node are retained to preserve the integrity of the reputation. Concurrently, endorsement relationships originating from or directed to the departing node are revoked. This reflects the principle that trust relations should not persist for inactive entities and ensures that no endorsement-based influence can propagate from a node that is no longer active.

This design prevents the exploitation of dormant endorsements in future reputation updates. Also, by preserving deactivated identities, we prevent nodes re-entering the network under new identities to escape past negative behaviors. If a node attempts to rejoin the network, its previous identity can be reactivated or linked, thereby maintaining continuity in reputation evolution. The reactivated node may regain reputation gradually, but it cannot bypass the accountability mechanism, as detailed in Sec.~\ref{subsec:impact of interactions}.

\subsection{Impact of Interactions}
\label{subsec:impact of interactions}
The novelty of our model is the fine-grained impact that interactions have on the network.
In particular, a positive (resp. negative) interaction feedback from node $i$ to $j$ affects both the endorsement confidence of the endorsement relationship and the entire reputation network by ``backwards propagating'' these rewards (and penalties) to $j$'s endorsers recursively.

In our model, every single interaction between $i$ and $j$ can result in positive or negative feedback scores from $i$ to $j$ that are in $[-10,10]$. These scores are then aggregated to update the cumulative values $p_{ij}$ and $n_{ij}$, which are subsequently used in the trustworthiness calculations with Eq.~\eqref{eq:trustworthiness} and~\eqref{eq:normalized_trustworthiness}. 

\noindent
\textbf{Penalties for Negative Interactions.} If a node $i$ experiences a negative interaction with node $j$, penalties are applied as follows:

\noindent
\textit{Step 1: Endorsement Penalty Signal.}
To incorporate the impact of negative feedback on endorsements, we first quantify the aggregated negative feedback for each node. The total negative feedback score for a node $j$, denoted $\mathcal{N}_j$, is the sum of all negative feedback scores it has received: $\mathcal{N}_j = \sum_i n_{ij}$. These scores constitute a vector $\mathcal{N} \in \mathbb{R}^N$.
An \textit{endorsement penalty signal} is defined as:
$g(\mathcal{N}_{j}) = e^{-\beta\cdot \mathcal{N}_{j}}, \beta > 0,$
where $\beta$ is the penalty sensitivity coefficient. 
This function ensures that the penalty signal decreases exponentially as $\mathcal{N}_{j}$ increases.

\noindent
\textit{Step 2: Backward Endorsement Penalty Propagation.}
The endorsers of misbehaving nodes are penalized via the BEPP accountability mechanism. 
Let \( \pi \in \mathbb{R}^{N} \) denote the endorser penalty vector, where \( \pi_i \) represents the total penalty to node \( i \) due to the misbehavior of its direct and indirect endorsees. 
The penalty is propagated upstream through the endorsement chain, discounting at each hop.
We calculate the penalty via a multi-hop recursive formulation:
\begin{small}
\begin{align*}
\pi 
&= \gamma E (\mathbf{1}-g(\mathcal{N})) + \gamma^2 E^2 (\mathbf{1}-g(\mathcal{N})) + \gamma^3 E^3 (\mathbf{1}-g(\mathcal{N})) + \cdots \\
&= \sum_{k=1}^{K} \gamma^k E^k (\mathbf{1}-g(\mathcal{N}))  = \gamma E (I - \gamma E)^{-1} (\mathbf{1}-g(\mathcal{N})), \numberthis
\label{eq:endorser penalty}
\end{align*}
\end{small}
where \( \gamma \in (0,1) \) is the discount factor that controls the decay of the penalty with increasing distance from the direct endorsers. The index \( k \) denotes the number of hops in the endorsement network over which the penalty is propagated. 
This formulation assigns greater responsibility to direct endorsers while diminishing the influence of indirect ones.
We introduce a convergence threshold \( \delta \) to stop the propagation when penalty becomes negligible. Specifically, if \( \pi^{(k)} \) denotes the penalty vector with \( k \) hops, the propagation is terminated once:
$\|\pi^{(k+1)} - \pi^{(k)}\|_1 < \delta,$
where \( \delta \) is a small positive constant. 
Alternatively, the propagation process can be terminated once a predefined maximum propagation bound $K$ is reached, ensuring that the propagation remains computationally feasible. 
This also applies to the backward endorsement reward propagation.

\noindent
\textbf{Rewards for Positive Interactions.}
While penalizing endorsers for poor endorsements improves accountability, it is equally important to reward those who endorse well-performing nodes. We now introduce the backward endorsement reward propagation (BERP) mechanism, designed to incentivize nodes to endorse those who are likely to bring positive impact to the network.
If node $i$ experiences a positive interaction with node $j$, the rewards are applied to both $j$ and its endorsers:

\noindent
\textit{Step 1: Endorsement Reward Signal.}
We use $\mathcal{P} \in \mathbb{R}^{N}$ to denote the accumulated positive feedback score vector, $\mathcal{P}_{j} = \sum_{i} p_{ij}$ denotes the accumulated positive feedback score of node $j$.
The \textit{endorsement reward signal} is defined as:
$r(\mathcal{P}_{j}) = 2 - e^{-\lambda \cdot \mathcal{P}_j}, \lambda > 0$,
where \( \lambda \) is a reward sensitivity coefficient. 

\noindent
\textit{Step 2: Backward Endorsement Reward Propagation.}
The endorsers of the well-performing nodes are rewarded by the BERP incentive mechanism. 
Let \( \rho \in \mathbb{R}^{N} \) denote the reward vector, where \( \rho_i \) represents the total reward to node \( i \) because of the honest behavior of its endorsees. The endorsement reward signal is propagated upstream through the endorsement network, discounting at each hop. Let \( r(\mathcal{P})^{\top} \in \mathbb{R}^N \) be the endorsement reward signal, the recursive backward endorsement reward vector \( \rho \in \mathbb{R}^N \) is:
\begin{small}
\begin{align*}
\rho 
&= \gamma E (r(\mathcal{P})-\mathbf{1}) + \gamma^2 E^2 (r(\mathcal{P})-\mathbf{1})  + \gamma^3 E^3 (r(\mathcal{P})-\mathbf{1}) + \cdots \\
&= \sum_{k=1}^{\infty} \gamma^k E^k (r(\mathcal{P})-\mathbf{1}) = \gamma E (I - \gamma E)^{-1} (r(\mathcal{P})-\mathbf{1}), \numberthis
\label{eq:endorser rewards}
\end{align*}
\end{small}
where \( \gamma \in (0,1) \) is the discount factor that controls the decay of the reward with increasing distance from the direct endorsers.

\noindent
\textbf{Endorsement Updating.}
After the penalties and rewards are calculated, RepuLink then updates the endorsement confidence from endorsers towards endorsees. Ideally, the endorsement can be manually adjusted by network users themselves, or updated based on the endorsement penalty/reward signals. Here, we will focus on the second one and introduce the automatic update strategy. The adjusted endorsement matrix $\hat{E}$ is computed as:
$\hat{E}= \hat{E} \circ (\mathbf{1} \cdot (g(\mathcal{N}) \circ r(\mathcal{P}))^\top)$,
where $\mathbf{1}$ is an $N \times 1$ column vector of ones, $(g(\mathcal{N}) \circ r(\mathcal{P}))^\top$ is the transpose of $g(\mathcal{N}) \circ r(\mathcal{P})$, and $\circ$ denotes the element-wise product. This operation effectively means each element $\hat{E}_{ij}$ is updated as $\hat{E}_{ij} \cdot g(\mathcal{N}_j) \cdot r(\mathcal{P}_j)$.
Subsequently, the updated endorsement matrix $\hat{E}$ is re-normalized using Eq.~\eqref{eq:normalised endorsement}, and then used for reputation updating. 

With the updated endorsement matrix \( E \) established, and the endorser penalty vector \( \pi \) and reward vector \( \rho \) defined, all necessary components for determining the final reputation for the subsequent time slot are available. 

\subsection{Reputation Updating}
\label{subsec:Reputation_Updating} 

\begin{algorithm}[t]
\caption{Iterative Reputation Update}
\label{alg:iterative_reputation_update}
\begin{algorithmic}[1] 
    \STATE \textbf{Input:} $R^{(t)}$, $T$, $E$, $\mathcal{P}$, $\mathcal{N}$
    \STATE \textbf{Parameters:} $\alpha$, $c$
    \STATE \textbf{Output:} $R^{(t+1)}$

    \STATE $\hat{T} \gets \text{Trustworthiness}$ \hfill \text{via Eq.~\eqref{eq:trustworthiness}}
    \STATE $T \gets \text{Normalized Trustworthiness}(\hat{T}, c)$ 
    \hfill \text{via Eq.~\eqref{eq:normalized_trustworthiness}}
    \STATE $g(\mathcal{N}),r(\mathcal{P}) \gets \text{Endorsement Penalty/Reward Signals}$
    \STATE $\pi,\rho \gets \text{Endorsement Penalty/Reward} $  \hfill \text{via Eqs.~\eqref{eq:endorser penalty}~\eqref{eq:endorser rewards}}
    \STATE $E \gets \text{Endorsement Updating}(g(\mathcal{N}),r(\mathcal{P}))$ 
    \STATE $\hat{R}^{(t+1)} \gets R^{(t)} - \pi + \rho \quad \text{Apply Penalty/Reward}$
    \STATE $\tilde{R}^{(t+1)} \gets \text{Reputation Updating} (T,E,\hat{R}^{(t+1)})$  \hfill \text{via Eq.~\eqref{eq:reputation}}
    \STATE $R^{(t+1)} \gets \text{Reputation Normalization}(\tilde{R}^{(t+1)}, c)$  \hfill \text{via Eq.~\eqref{eq:reputation-normalization}}
    \RETURN $R^{(t+1)}$
\end{algorithmic}
\end{algorithm}

The iterative reputation update process (as illustrated in Alg.~\ref{alg:iterative_reputation_update}) starts by reassessing and normalizing pairwise trustworthiness \(T\) from the latest interaction feedback (Lines 4,5 in Alg.~\ref{alg:iterative_reputation_update}). 
Subsequently, a backward propagation mechanism refines the intermediate reputation by applying penalties \( \pi \) and rewards \( \rho \) to yield an adjusted vector \( \hat{R}^{(t+1)}\) (Lines 6,7,9 in Alg.~\ref{alg:iterative_reputation_update}).
The updated trustworthiness \(T\), along with the updated endorsement matrix \(E\), then informs a forward propagation step to compute the reputation vector \( \tilde{R}^{(t+1)} \) (Lines 5,8,10 in Alg.~\ref{alg:iterative_reputation_update}).  
Finally, the reputation vector \( \tilde{R}^{(t+1)} \) is normalized (Line 11 in Alg.~\ref{alg:iterative_reputation_update}) to produce the definitive and consistently-scaled reputation scores \( R^{(t+1)} \) for the next timeslot. 

\begin{table}[t]
\centering
\caption{Summary statistics of trust datasets}
\label{tab:datasets}
\resizebox{0.8\columnwidth}{!}{%
\begin{tabular}{|l|cccc|}
\bottomrule
\textbf{Dataset}     & \textbf{\#Nodes} & \textbf{\#Edges} & \textbf{Range} & \textbf{\% Positive} \\
\toprule \hline
Bitcoin OTC          & $5,881$            & $35,592$           & $[-10, 10]$     & $89\%$          \\
Bitcoin Alpha        & $3,783$            & $24,186$           & $[-10, 10]$       & $93\%$        \\
Epinions             & $75,879$          & $508,837$          &    $/$     &     $/$    \\
Synthetic             & $5,000$          & $50,006$          &    $/$     &     $/$    \\
\toprule
\end{tabular}}
\end{table}

\section{Experiments}
\label{sec:experiments}
This section evaluates the performance of our proposed reputation model under various experimental settings on real-world datasets and compares our results against state-of-the-art trust models (see Sec.~\ref{sec:introduction} and~\ref{sec:related-works}) presented in Tab.~\ref{tab:role-specific-pagerank}. 
All experiments were conducted on an Apple machine equipped with an \textit{M1 Max} chip and 32\,GB of RAM. One comparative model, \textit{ShapleyTrust}, exhibited excessive running times and consistently timed out on this setup. To ensure its termination for a complete correctness comparison, \textit{ShapleyTrust} was specifically executed on a high-performance computing cluster node featuring an \textit{Intel Xeon Gold 6138 CPU} and 80\,GB of RAM.

\subsection{Experimental Setup}
\label{subsec:experimental_setup}
\textbf{Datasets.}
We conduct experiments on two public trust datasets derived from peer-to-peer marketplaces: \textit{Bitcoin-OTC} and \textit{Bitcoin-Alpha}~\cite{kumar2016edge,kumar2018rev2} that have been used as benchmark datasets for related work in this problem.
Bitcoin-OTC~\footnote{https://snap.stanford.edu/data/soc-sign-bitcoin-otc.html} is a directed, weighted network constructed from an over-the-counter trading platform, where each edge represents a subjective trust rating similar to our interaction feedback. 
Bitcoin-Alpha~\footnote{https://snap.stanford.edu/data/soc-sign-bitcoin-alpha.html} exhibits a similar structure but originates from a different Bitcoin forum-based marketplace.
As these two datasets do not include explicit endorsement relationships, we incorporate the \textit{Epinions} dataset~\cite{richardson2003trust} to simulate the \emph{endorsement network}. The Epinions~\footnote{https://snap.stanford.edu/data/soc-Epinions1.html} dataset is suitable for simulating endorsements as it is a who-trust-whom online social network of a general consumer review site Epinions.com.
To construct the Layer 1 endorsement network, we align the Bitcoin datasets with the Epinions network by mapping nodes based on shared user identities.
Epinions edges are unweighted, so their initial endorsement confidence value is set to $1$. This value serves as a basis for initial reputation assignments and will be updated based on node interactions.

\noindent
\textbf{Competitor Approaches.}
We compare against five trust models discussed in Sec.~\ref{sec:related-works} (also Tab.~\ref{tab:role-specific-pagerank} in Sec.~\ref{sec:introduction})
mainly utilizing interaction networks for trust propagation and evaluation. Our implementations adhere to the methodologies presented in the original publications. 
Notably, since exact ShapleyTrust calculation is computationally intractable, we employ a polynomial sampling approach~\cite{castro2009polynomial} with a sample size of 500 for Shapley value estimation, balancing accuracy with computational feasibility.

\noindent
\textbf{Evaluation Metrics.} We select four widely-used metrics as below:

\emph{(1) Area Under the ROC Curve (AUC).} Assesses overall ranking quality globally, as it measures the ability of a scoring function to correctly rank high-reputation users above low-reputation users. It is computed based on the Receiver Operating Characteristic (ROC) curve, which plots the True Positive Rate (TPR) versus the False Positive Rate (FPR) at varying thresholds. Using the trapezoidal rule for approximation, it is calculated as:
$\text{AUC} \approx \sum_{i=1}^{n-1} \frac{(\text{FPR}_{i+1} - \text{FPR}_i)(\text{TPR}_{i+1} + \text{TPR}_i)}{2}$.
The area under the ROC curve quantifies the ranking quality, where AUC = 1 indicates perfect discrimination, and AUC = 0.5 denotes random guessing.

\emph{(2) Precision@K (PK).} Focuses on the practical effectiveness of identifying the most reputable users at the top of the list. It evaluates the proportion of true high-reputation users retrieved among the top-$K$ ranked users. 
$\text{Precision@}K = \frac{|\text{Top-}K \cap \text{High-Reputation Users}|}{K}$.
This metric emphasizes the algorithm's accuracy in identifying high-reputation users within the highest-ranked positions.

\emph{(3) Kendall's Tau Rank Correlation (KT).} Evaluates the consistency of the predicted rankings compared to the ground truth, measuring the ordinal association between the predicted ranking and the ground-truth ranking. It is defined as:
$\tau = \frac{n_c - n_d}{\frac{1}{2}n(n - 1)}$
where $n_c$ is the number of concordant pairs, $n_d$ is the number of discordant pairs, and $n$ is the total number of nodes. A value of $\tau = 1$ indicates perfect agreement, while $\tau = 0$ implies no correlation (random order).

\emph{(4) Spearman's Rank Correlation Coefficient (SRC).} Measures the strength and direction of the monotonic relationship between the predicted reputation scores (or ranks) and the ground-truth ranks. It is computed as:
$\rho = 1 - \frac{6 \sum_{i=1}^{n} d_i^2}{n(n^2 - 1)}$
where $d_i$ is the difference between the predicted rank and the true rank for node $i$, and $n$ is the total number of nodes. A higher $\rho$ (closer to 1) indicates a stronger positive monotonic correlation in ranking, while $\rho$ closer to -1 indicates a stronger negative monotonic correlation.

\begin{figure*}
    \centering
    \includegraphics[width=0.9\linewidth]{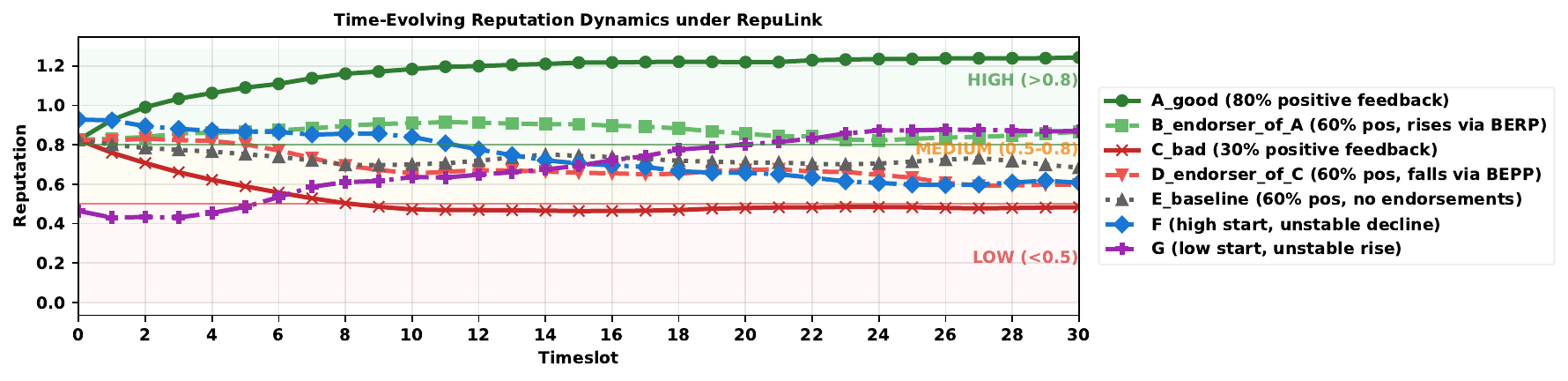}
    \caption{Time-evolving reputation dynamics under RepuLink over $30$
    time slots. Display values are sqrt-compressed so the uniform
    initial state maps to $0.8$; tiers are HIGH ($>0.8$), MEDIUM
    ($0.5$--$0.8$), LOW ($<0.5$).}
    \label{fig:reputation_dynamics}
\end{figure*}

\begin{figure*}[ht]
    \centering
    \begin{subfigure}{0.3\linewidth}
        \includegraphics[width=\linewidth]{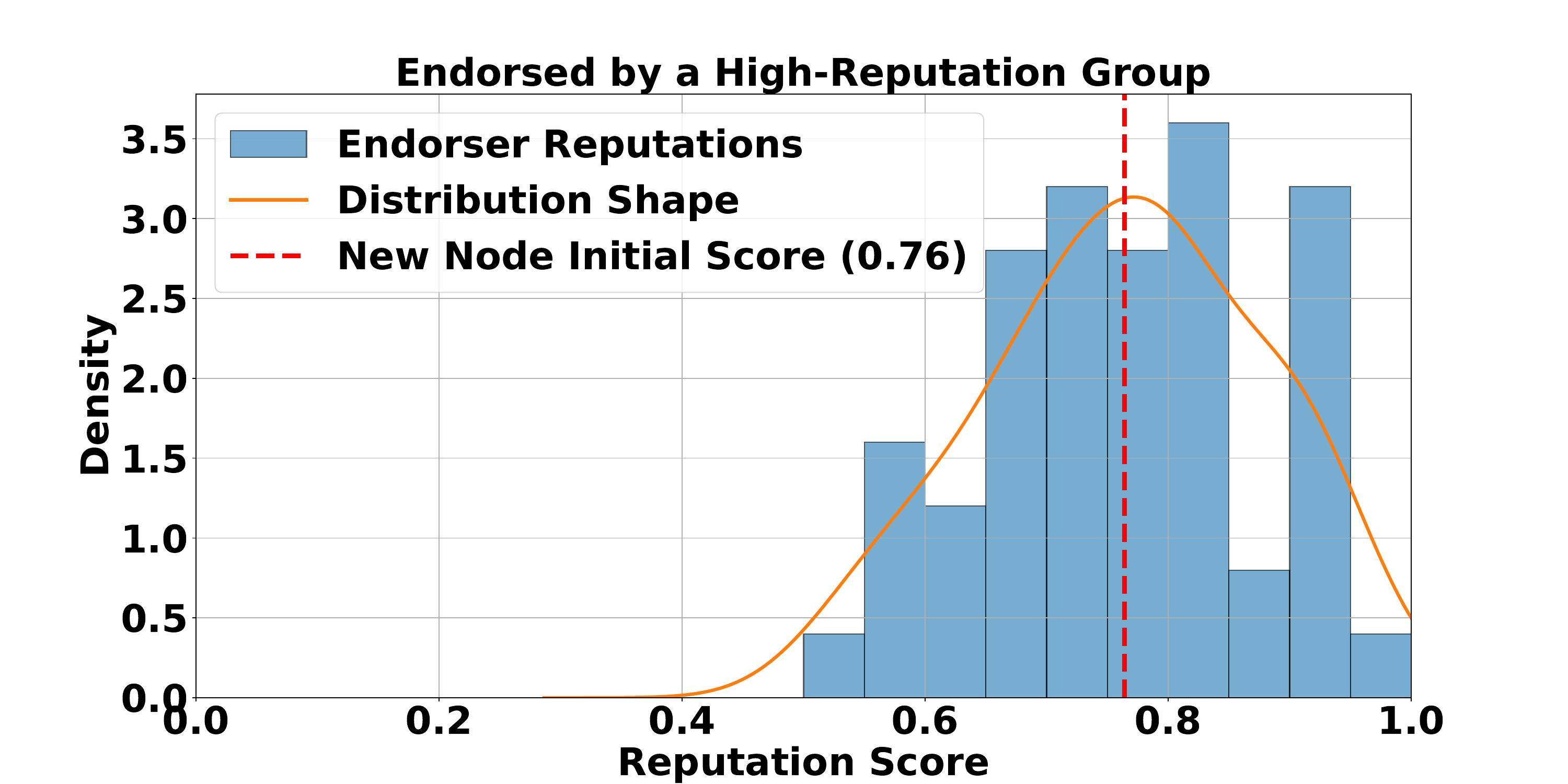}
    \end{subfigure}
    \hfill
    \begin{subfigure}{0.3\linewidth}
        \includegraphics[width=\linewidth]{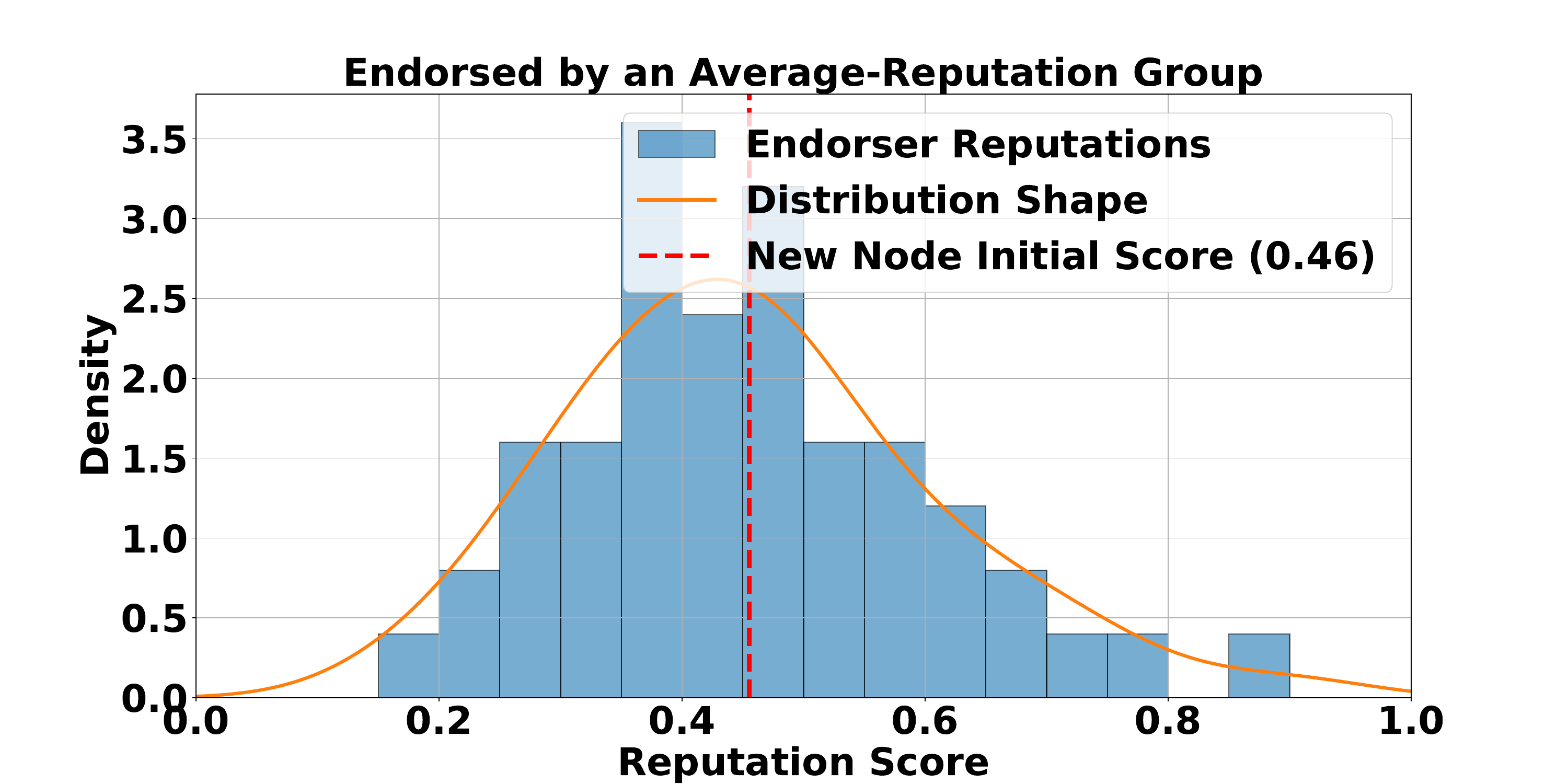}
    \end{subfigure}
    \hfill
    \begin{subfigure}{0.3\linewidth}
        \includegraphics[width=\linewidth]{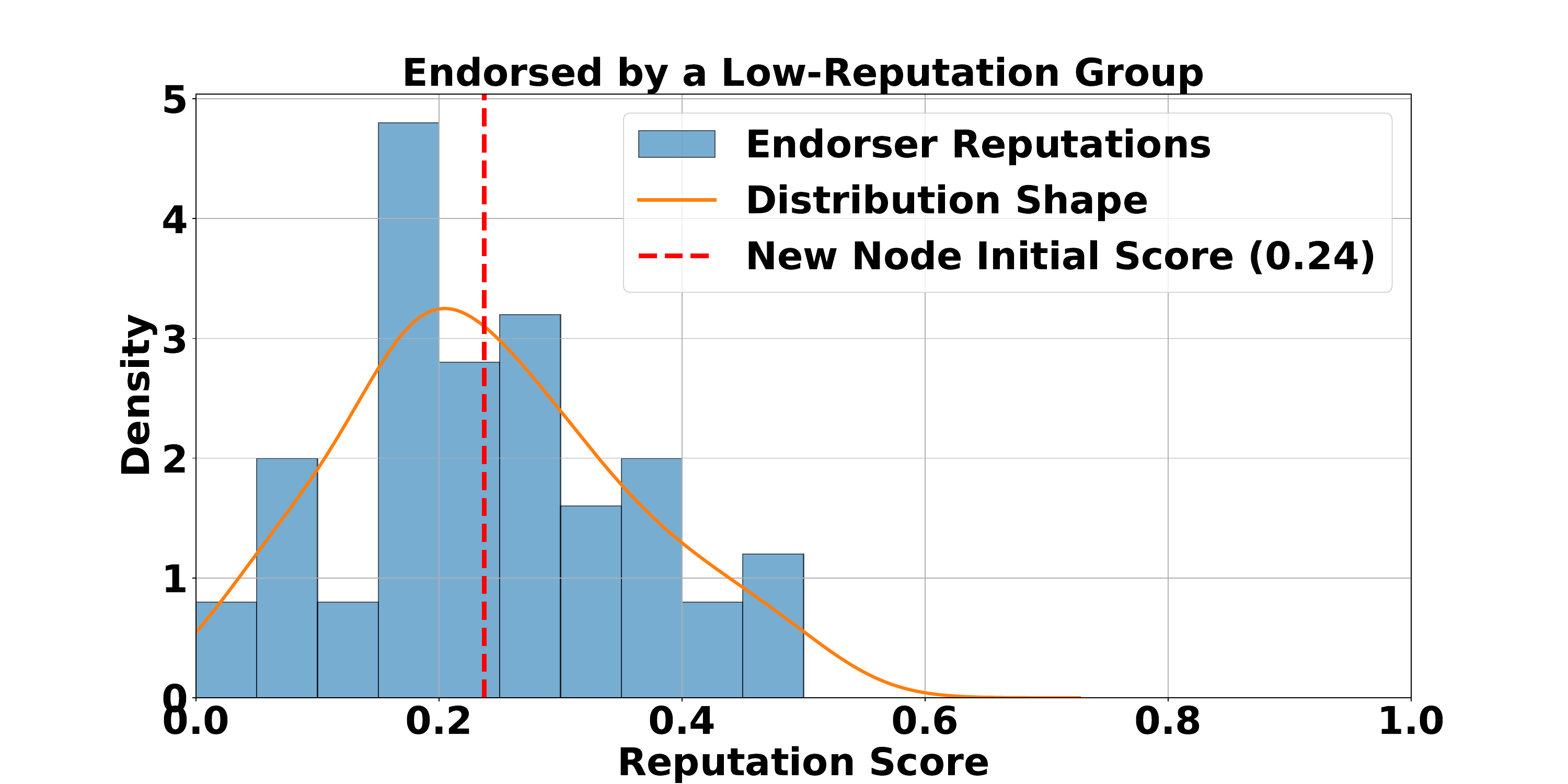}
    \end{subfigure}
    \caption{Demonstration of initial reputation assignment for a new node under different endorsement scenarios.}
    \label{fig:endorser_distribution}
\end{figure*}

\subsection{Dynamic Reputation Evolution}
\label{subsec:dynamic_score_evolution}
To illustrate RepuLink's two-layer mechanism in action, we design a controlled case study with seven synthetic nodes over $30$ time slots.
Each time slot draws a randomized number of pairwise interactions (mean $50$, std $14$), with per-node positive-rating probabilities jittered by Gaussian noise to emulate day-to-day fluctuation.
The endorsement matrix $E$ is held fixed across time slots so that the effects of Backward Endorsement Reward Propagation (BERP) and Penalty Propagation (BEPP) are attributable solely to the backward step; otherwise, endorsers would gradually withdraw their edges and the signal would self-extinguish.
There are seven nodes with various properties: well-behaved (A), misbehaving (C), isolated endorsers of A and C (B,~D), a no-endorsement control (E), and two time-varying actors (F~falling, G~rising).

Figure~\ref{fig:reputation_dynamics} reveals three features of RepuLink.
First, forward propagation clearly separates good and bad actors: A rises to the HIGH tier while C collapses into LOW.
Second, the backward step induces a \emph{symmetric divergence} between B, D, and the control E, all of which share the same intrinsic $60\%$ positive rate. B rises to $\approx 0.87$ because BERP flows reward through its endorsement of A, while D falls to $\approx 0.60$ because BEPP flows penalty through its endorsement of C.
Since B, D, and E are identical apart from their endorsement choices, the gap isolates the contribution of backward propagation.
Third, the time-varying actors F and G track gradual behavioral drift. F decays from a high initial reputation ($0.90$) into MEDIUM, while G climbs from LOW ($0.45$) to the upper edge of MEDIUM ($\approx 0.87$), demonstrating that the cumulative-counts formulation remains responsive to underlying behavioral change without over-reacting to per-time-slot noise.

\subsection{Effectiveness Evaluation}
\label{subsec:effectiveness evaluation}

\noindent
\textit{\textbf{Cold-Start Scenario}}
This experiment is designed to demonstrate RepuLink's ability to address the cold-start problem by leveraging its social endorsement layer to assign a meaningful initial reputation.
We use the \textit{Epinions} dataset to build an initial trust network. We then demonstrate the initial reputation assignment for three newly joined nodes mainly endorsed by a high-reputation group, an average-reputation group, and a low-reputation group, respectively.

\begin{table*}[t]
\centering
\caption{Performance Comparison: RepuLink vs. Others}
\label{tab:performance_comparison}
\resizebox{0.85\textwidth}{!}{%
\begin{tabular}{|c|cccc|cccc|cccc|cccc|}
\bottomrule
\multirow{3}{*}{\textbf{Algorithm}}  & \multicolumn{8}{c|}{\textbf{Bitcoin-OTC}}                        & \multicolumn{8}{c|}{\textbf{Bitcoin-Alpha}}                       \\ \cline{2-17}
& \multicolumn{4}{c|}{\textbf{Layer 2}}                        & \multicolumn{4}{c|}{\textbf{Layer 1\&2}} & \multicolumn{4}{c|}{\textbf{Layer 2}}                        & \multicolumn{4}{c|}{\textbf{Layer 1\&2}}                      \\
 &
  \multicolumn{1}{c}{\textbf{AUC}} &
  \multicolumn{1}{c}{\textbf{PK}} &
  \multicolumn{1}{c}{\textbf{KT}} &
  \multicolumn{1}{c|}{\textbf{SRC}} &
    \multicolumn{1}{c}{\textbf{AUC}} &
  \multicolumn{1}{c}{\textbf{PK}} &
  \multicolumn{1}{c}{\textbf{KT}} &
  \multicolumn{1}{c|}{\textbf{SRC}} &
  \multicolumn{1}{c}{\textbf{AUC}} &
  \multicolumn{1}{c}{\textbf{PK}} &
  \multicolumn{1}{c}{\textbf{KT}} &
  \multicolumn{1}{c|}{\textbf{SRC}} &
  \multicolumn{1}{c}{\textbf{AUC}} &
  \multicolumn{1}{c}{\textbf{PK}} &
  \multicolumn{1}{c}{\textbf{KT}} &
  \multicolumn{1}{c|}{\textbf{SRC}} \\ \toprule \hline
PageRank\cite{page1999pagerank}                            
& 0.74          & 0.68          & 0.34          & 0.41 
& 0.69 & 0.69 & 0.30 & 0.37 
& \underline{0.74}          & \underline{0.69}          & \underline{0.33}          & \underline{0.41}   
& \underline{0.77} & \underline{0.72} & \underline{0.38} & \underline{0.46}  \\
EigenTrust\cite{kamvar2003eigentrust}                          
& \underline{0.75}          & 0.71          & \underline{0.35}          & \underline{0.43}    
& 0.68 & 0.67 & 0.24 &   0.29    
& 0.67          & 0.62          & 0.27          & 0.33
& 0.76 & 0.71 & 0.37 & 0.45   \\
PowerTrust\cite{zhou2007powertrust}                         
& 0.67          & \underline{0.73}          & 0.24          & 0.28    
& 0.72 & 0.69  & 0.32 &   0.39   
& 0.61          & 0.56          & 0.15          & 0.18
& 0.69 & 0.62 & 0.27 & 0.33     \\
AbsoluteTrust\cite{awasthi2020absolutetrust}                       
& 0.46          & 0.46          & -0.06         & -0.07
& 0.65 & 0.66 & 0.22 &  0.27    
& 0.56          & 0.54          & 0.09          & 0.11 
& 0.76 & 0.71 & 0.37 & 0.45   \\
ShapleyTrust\cite{bandhana2024trust}                       
& 0.73          & 0.71          & 0.33          & 0.40
& \underline{0.73} & \underline{0.71} & \underline{0.33} &   \underline{0.41}   
& 0.68          & 0.64          & 0.25          & 0.31
& 0.69 & 0.65 & 0.26 & 0.32    \\
\textbf{RepuLink}                   
& \textbf{0.83} & \textbf{0.75} & \textbf{0.46} & \textbf{0.56} 
& \textbf{0.81} & \textbf{0.73} & \textbf{0.44} & \textbf{0.53}  
& \textbf{0.84} & \textbf{0.77} & \textbf{0.47} & \textbf{0.58} 
& \textbf{0.85} & \textbf{0.76} & \textbf{0.49} &  \textbf{0.60}  \\ \hline
Improvement 
\% & 8\% & 2\% & 11\% & 13\% 
& 8\% & 2\% & 11\% & 12\% 
& 10\% & 8\% & 14\% & 17\% 
& 8\% & 4\% & 11\% & 14\%\\
\toprule
\multicolumn{17}{l}{Note: In each column, the best-performing result is indicated in \textbf{bold}, while the improvement is computed relative to the \underline{best baseline}.}
\end{tabular}
}
\end{table*}

\noindent
\textit{- Performance Analysis.}
Each scenario involves a new node joining the network with a different endorsement profile—defined by the reputation distribution of the 50 existing nodes endorsing it. As shown in Fig.~\ref{fig:endorser_distribution}, the results demonstrate that RepuLink assigns an interpretable initial reputation that is a direct function of the endorser group's quality.
When the new node is endorsed mainly by high-reputation members, RepuLink aggregates this strong social trust and assigns a high initial reputation, allowing for immediate integration into the network.
When the endorser group's reputation follows a normal distribution, the resulting initial reputation is appropriately moderate, reflecting the average trustworthiness of its social backers.
When the endorser group consists mostly of low-reputation members, the calculated initial reputation is correctly low, signaling that the node's social connections are of questionable quality.
These outcomes validate the expected behavior of the endorsement mechanism as discussed in Sec.~\ref{subsec:node arrival and departure}.

\noindent
\textit{\textbf{Evaluation of RepuLink (Layer 2 only)}}
Tab.~\ref{tab:performance_comparison} presents the experimental results for RepuLink and all competitor models across the four evaluation metrics.
It is evident that the \textit{RepuLink} model consistently outperforms all other methods across the four representative metrics on both the Bitcoin-OTC/Alpha datasets. RepuLink's superior ability to identify high-reputation nodes and rank them ahead of low-reputation nodes is supported by its AUC results (0.83 and 0.84 on the two datasets, respectively).

The model's capacity to maintain high ranking quality (i.e., the percentage of correctly identified positive nodes within the top-\(K\) results) is substantiated by its leading Precision@K values of 0.75 (Bitcoin-OTC) and 0.77 (Bitcoin-Alpha). Although RepuLink's Precision@K performance on Bitcoin-OTC may not uniformly exceed that of other models for all values of \(K\), it achieves the highest performance as \(K\) increases. This indicates that our model effectively balances the ranking of very high-reputation nodes with strong overall performance across a broader range of network nodes.
Additionally, RepuLink's Kendall's Tau (KT) results of 0.46 and 0.47 on the two datasets demonstrate its enhanced performance in measuring the ordinal consistency between the predicted ranking and the ground-truth ranking. The SRC results (0.56 for Bitcoin-OTC and 0.58 for Bitcoin-Alpha) further indicate RepuLink's superior ability to preserve the monotonic relationship between its predicted reputation scores (or rankings) and the ground truth.

\noindent
\textit{\textbf{Evaluation of RepuLink (Full Model)}}
To evaluate the performance of RepuLink, we compared it with five other models.

\noindent
\textit{- Ground Truth.}
The ground truth generated for this experiment aims to identify nodes with high and low reputation by combining interaction history with endorsements. Initially, an interaction-based reputation is calculated for each node using the method for Layer 2 only experiment. Concurrently, the number of incoming endorsements for each node is counted and normalized. Then a combined score is computed for each node by taking a weighted average of its interaction rating and its normalized endorsements. 
Similarly, the top 20\% of users are labeled ``high-reputation'' and the bottom 20\% ``low-reputation''.
Since the competitor models do not incorporate an endorsement layer, we adapt their outputs to ensure a fair comparison. For each competitor model, its specific methodology is used to compute an interaction-based trustworthiness score. This score is then combined with the same normalized endorsement counts used for our ground truth generation, employing an identical weighted averaging approach. 

\begin{figure}[t]
    \centering
    \begin{subfigure}{0.43\linewidth}
        \includegraphics[width=\linewidth]{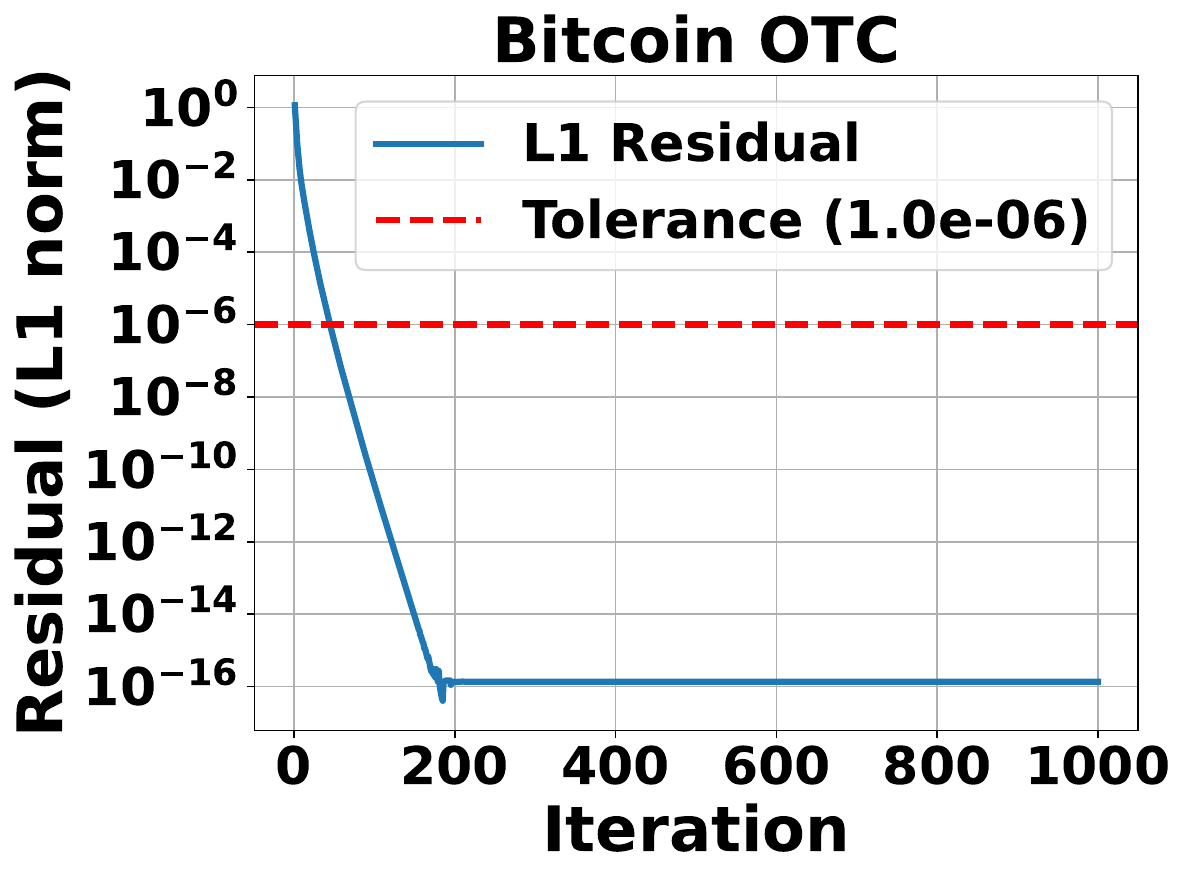}
    \end{subfigure}
    \hfill
    \begin{subfigure}{0.43\linewidth}
        \includegraphics[width=\linewidth]{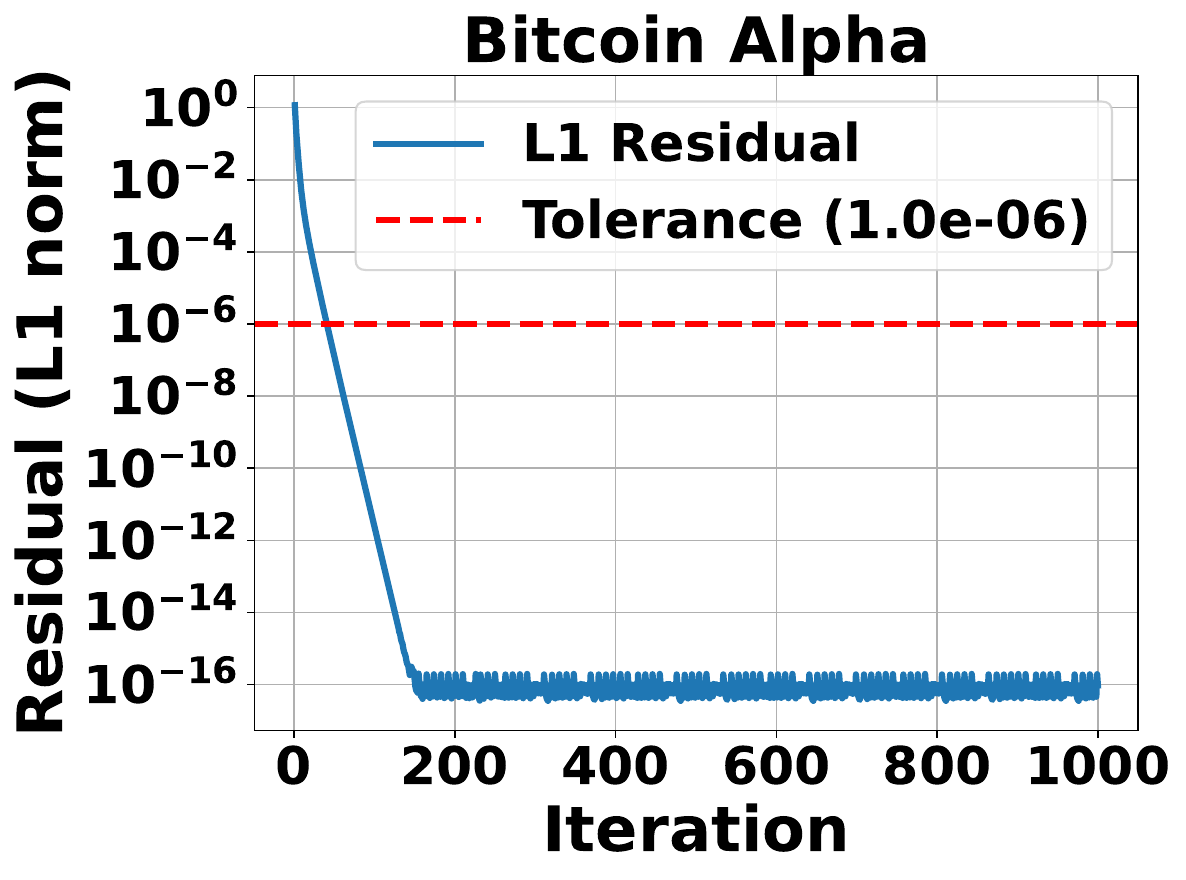}
    \end{subfigure}
    \caption{Convergence Curve of RepuLink.}
    \label{fig:convergence}
\end{figure}

\begin{table}[t]
\centering
\caption{Convergence Analysis: RepuLink vs. Others}
\label{tab:convergence_comparison}
\resizebox{0.85\columnwidth}{!}{%
\begin{tabular}{|c|ccc|ccc|}
\bottomrule
\multirow{2}{*}{\textbf{Algorithm}}  & \multicolumn{3}{c|}{\textbf{Bitcoin-OTC}}  & 
\multicolumn{3}{c|}{\textbf{Bitcoin-Alpha}}                   \\ \cline{2-7} 
 &
  \multicolumn{1}{c}{\textbf{\textcircled{1}}} &
  \multicolumn{1}{c}{\textbf{\textcircled{2}}} &
  \multicolumn{1}{c|}{\textbf{\textcircled{3}}} &
  \multicolumn{1}{c}{\textbf{\textcircled{1}}} &
  \multicolumn{1}{c}{\textbf{\textcircled{2}}}  &
  \multicolumn{1}{c|}{\textbf{\textcircled{3}}}\\ \toprule \hline
PageRank~\cite{page1999pagerank}       & 50         & 0.05          & $9e-7$     & 50         & 0.05          & $9e-7$    \\
EigenTrust~\cite{kamvar2003eigentrust}     & 47         & 0.29          & $9e-7$     & 47         & 0.29          & $9e-7$    \\
PowerTrust~\cite{zhou2007powertrust}     & 53         & 0.01         & $8e-7$     & 55         & 0.01          & $8e-7$    \\
AbsoluteTrust~\cite{awasthi2020absolutetrust}  & 39         & 2.73          & $9e-7$     & 39         & 2.72          & $9e-7$    \\
ShapleyTrust~\cite{bandhana2024trust}  & N/A         & N/A          & N/A     & N/A         & N/A          & N/A    \\
\textbf{RepuLink}  & 45     & 0.39          &  $9e-7$    & 45         & 0.39          & $9e-7$    \\
\toprule
\multicolumn{7}{l}{
\textcircled{1}: Number of iterations to the specified tolerance. } \\
\multicolumn{7}{l}{
\textcircled{2}: Running time in seconds. \textcircled{3}: Residual at the last iteration.}
\end{tabular}}
\end{table}

\noindent
\textit{- Performance Analysis.}
Tab.~\ref{tab:performance_comparison} presents the experimental results for RepuLink and competitor models.
Echoing its strong performance characteristics (such as the capacity to maintain high ranking quality) observed in the Layer 2 only scenario analysis, RepuLink in its comprehensive configuration consistently outperforms competitor models across the four metrics on both the Bitcoin-OTC/Alpha datasets.
Specifically, \textit{RepuLink} achieves AUC scores of 0.81 and 0.85 on the two datasets, respectively, underscoring its robust ability to distinguish between high- and low-reputation nodes. 
Regarding Precision@K, particularly on the Bitcoin-Alpha dataset when utilizing both layers, RepuLink's performance may not uniformly exceed that of other models for all values of \(K\), it demonstrates better performance as \(K\) increases.
Additionally, RepuLink's KT results of 0.44 (Bitcoin-OTC) and 0.49 (Bitcoin-Alpha) highlight its enhanced performance in measuring the ordinal consistency between its predicted rankings and the ground-truth ranking. The SRC results (0.53 for Bitcoin-OTC and 0.60 for Bitcoin-Alpha) further indicate RepuLink's better ability to preserve the monotonic relationship between its predicted reputation and the ground truth.

Tab.~\ref{tab:convergence_comparison} summarizes the convergence performance of \textit{RepuLink} and the competitor models, with Fig.~\ref{fig:convergence} further visualizing \textit{RepuLink}'s convergence trajectory.
While not necessarily the fastest overall, \textit{RepuLink} demonstrates efficient convergence, reaching the predefined tolerance threshold of $10^{-6}$ in approximately 45 iterations (0.39\,s runtime) on both datasets.
The convergence curves in Fig.~\ref{fig:convergence} affirm this robust behavior: the residual error drops rapidly within the initial tens of iterations, quickly satisfying the $10^{-6}$ benchmark and subsequently stabilizing at even lower levels.
\textit{ShapleyTrust}~\cite{bandhana2024trust} is listed as N/A in Tab.~\ref{tab:convergence_comparison} due to its non-iterative calculation method and its significant computational time, which required execution on an HPC cluster even when employing Shapley value approximation techniques~\cite{castro2009polynomial}.

\begin{figure*}[t]
    \centering
    \begin{subfigure}{0.2\linewidth}
        \includegraphics[width=\linewidth]{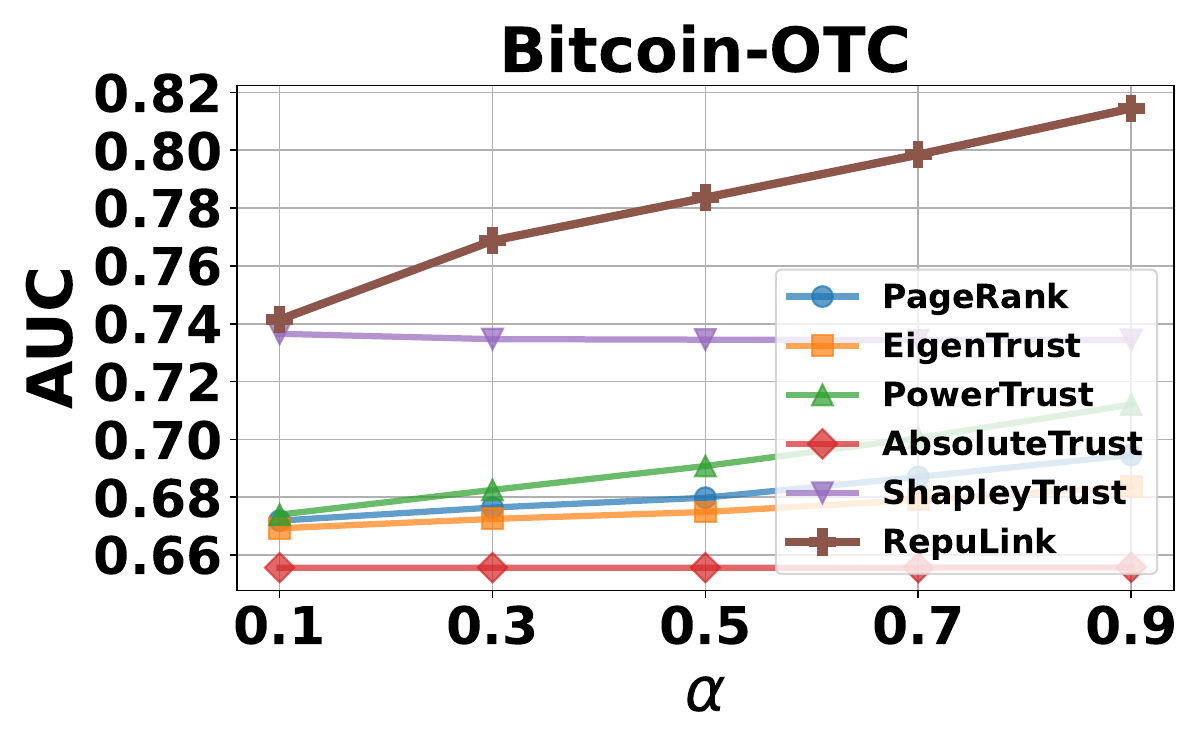}
    \end{subfigure}\hfill
    \begin{subfigure}{0.2\linewidth}
        \includegraphics[width=\linewidth]{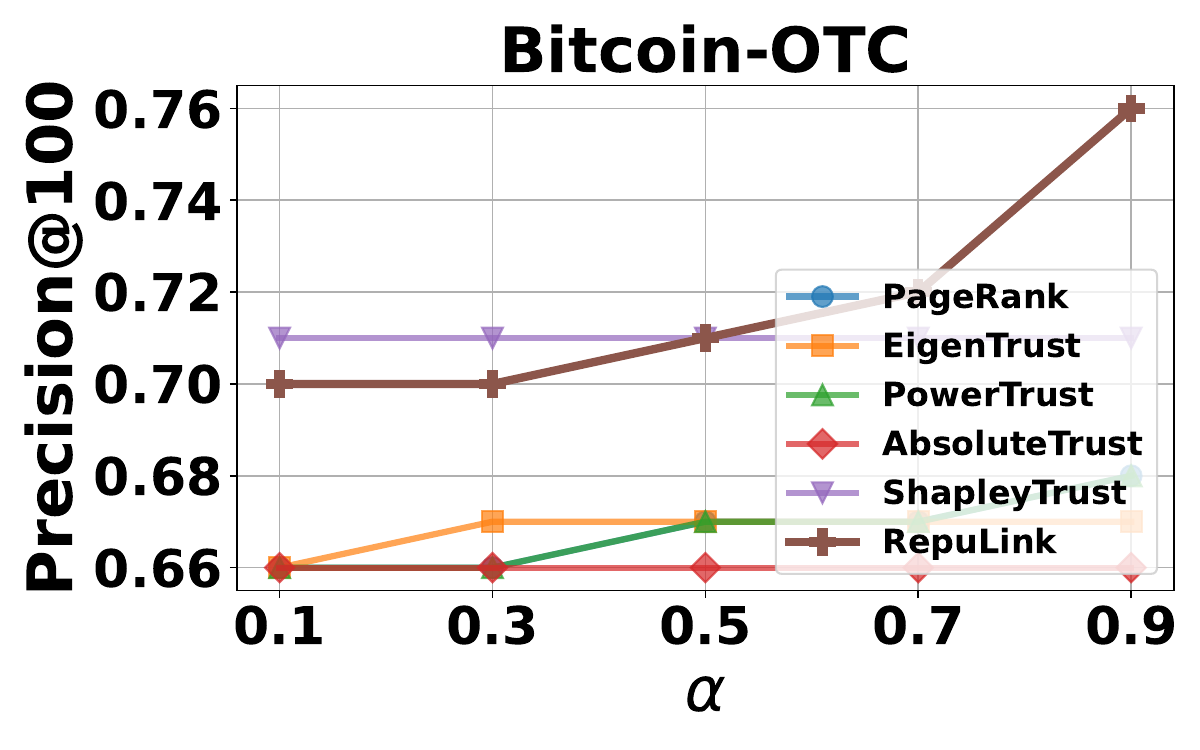}
    \end{subfigure}\hfill
    \begin{subfigure}{0.2\linewidth}
        \includegraphics[width=\linewidth]{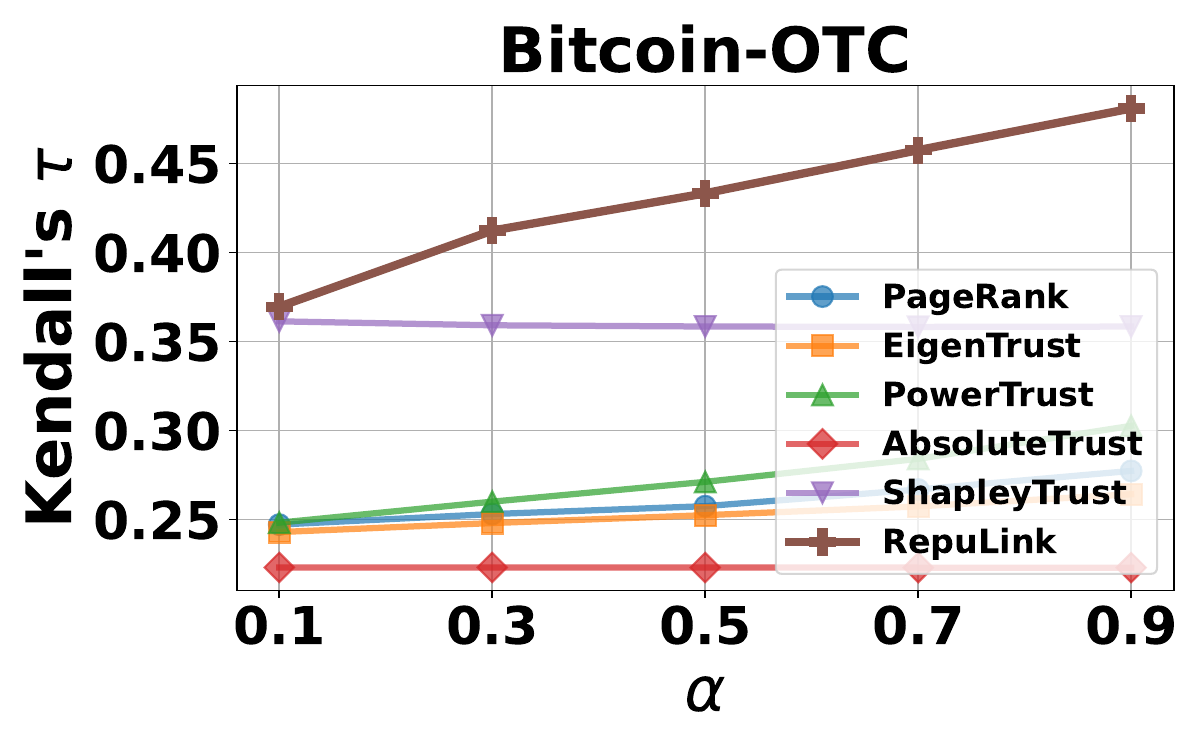}
    \end{subfigure}\hfill
    \begin{subfigure}{0.2\linewidth}
        \includegraphics[width=\linewidth]{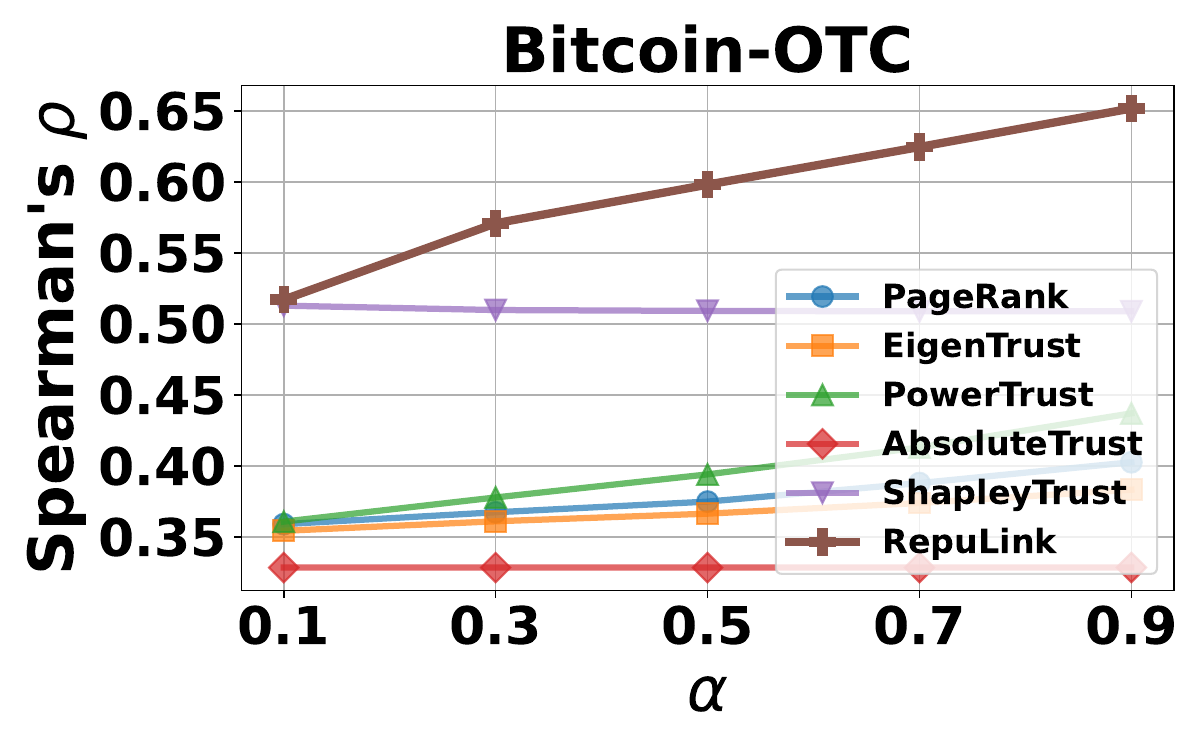}
    \end{subfigure}
    \\[2pt]
    \begin{subfigure}{0.2\linewidth}
        \includegraphics[width=\linewidth]{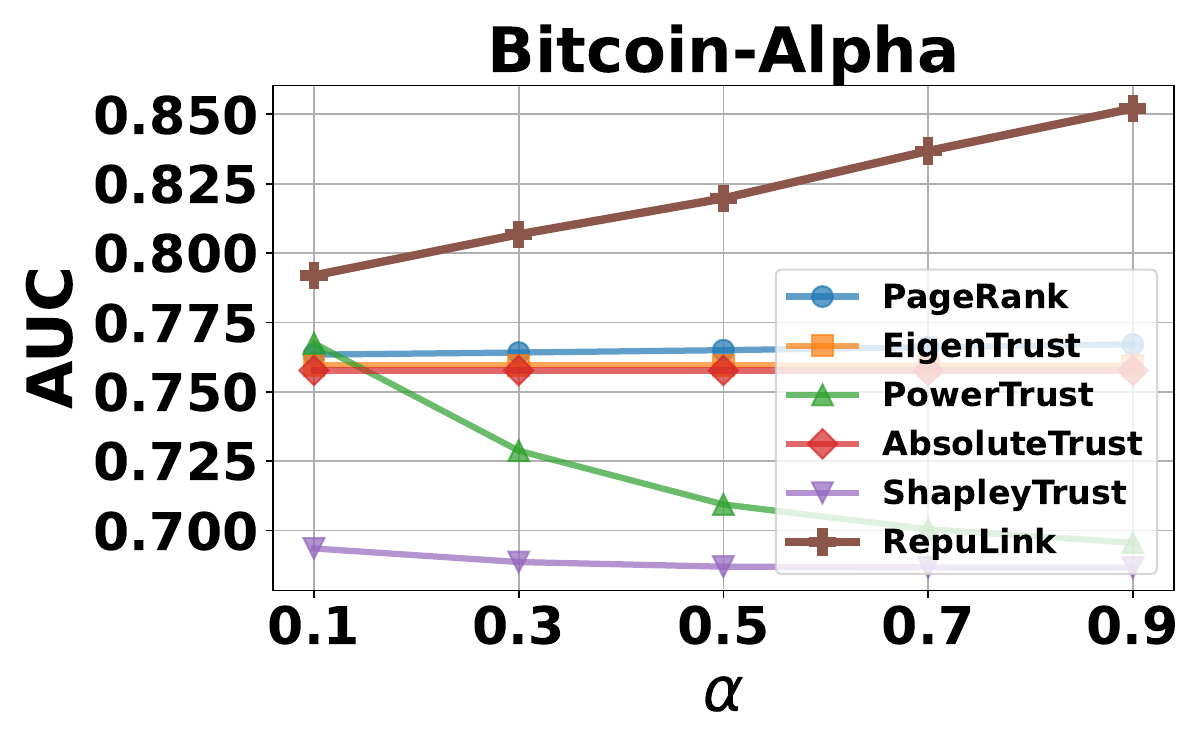}
    \end{subfigure}\hfill
    \begin{subfigure}{0.2\linewidth}
        \includegraphics[width=\linewidth]{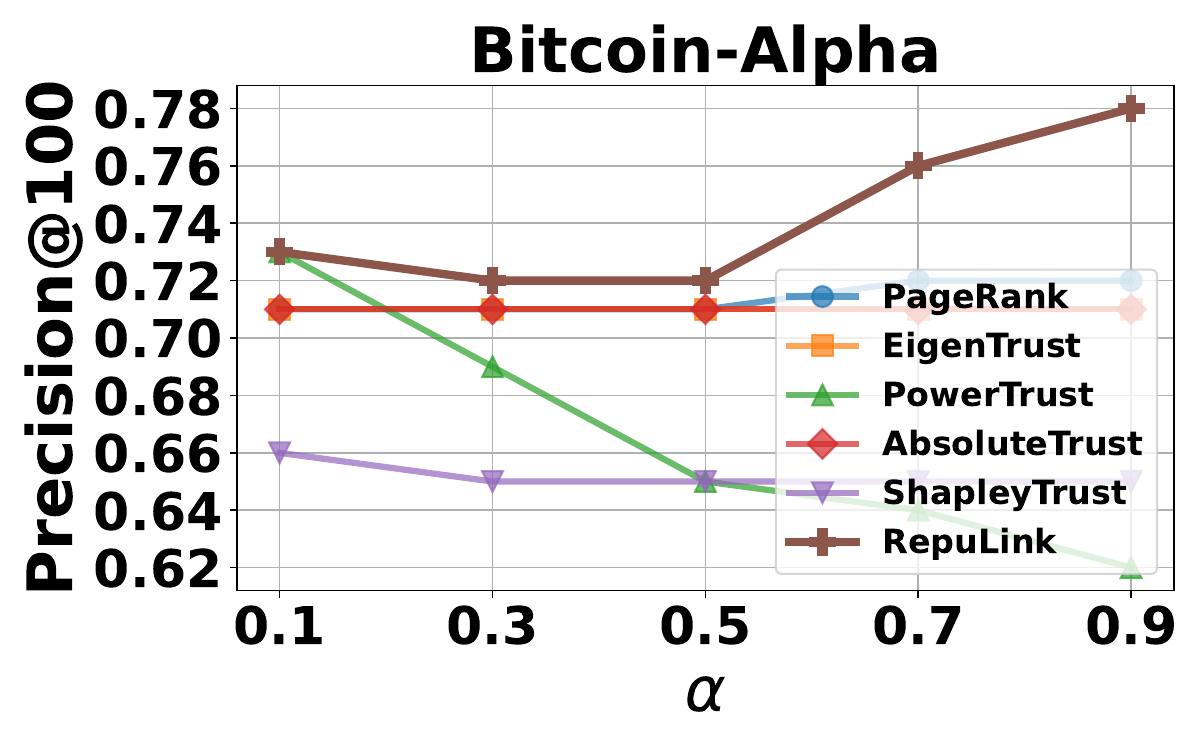}
    \end{subfigure}\hfill
    \begin{subfigure}{0.2\linewidth}
        \includegraphics[width=\linewidth]{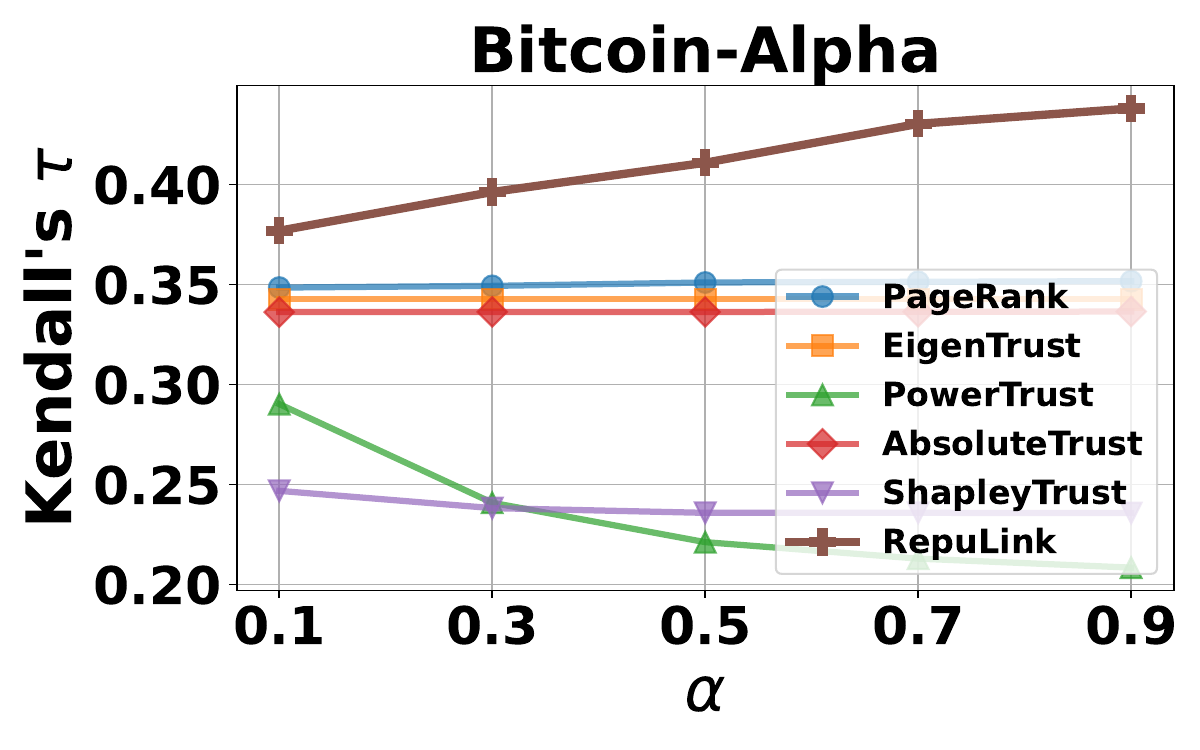}
    \end{subfigure}\hfill
    \begin{subfigure}{0.2\linewidth}
        \includegraphics[width=\linewidth]{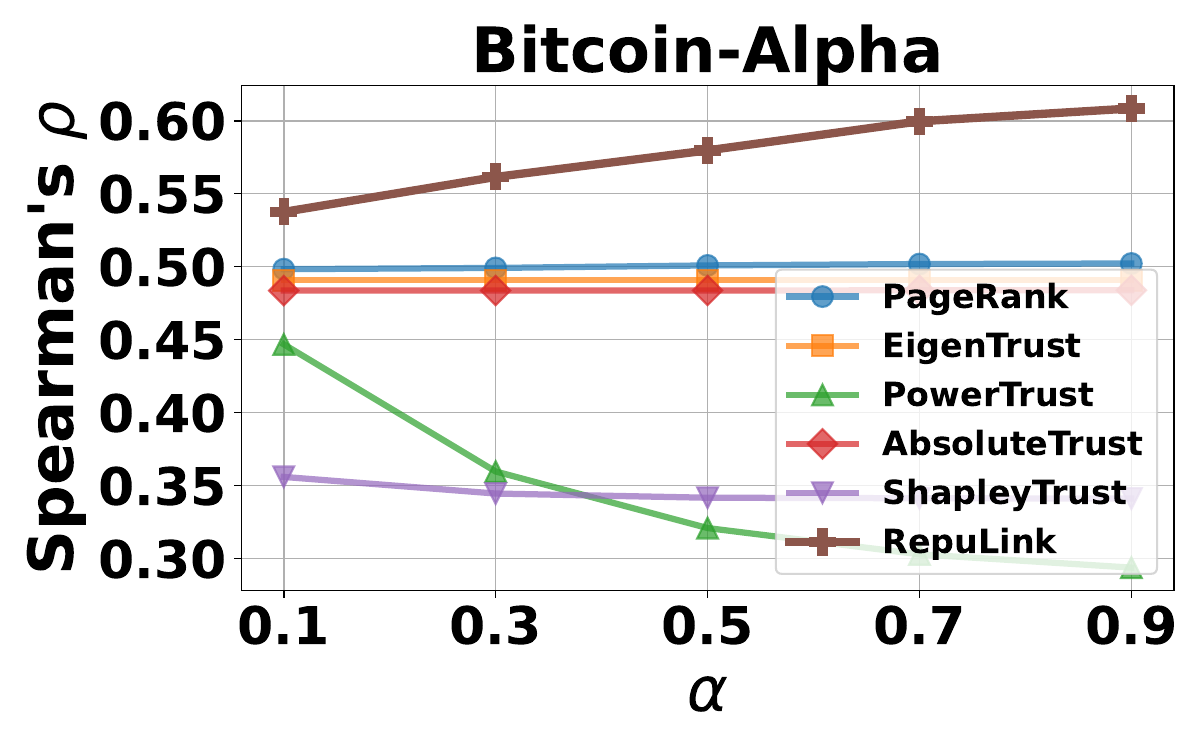}
    \end{subfigure}
    \caption{Parameter sensitivity on Bitcoin-OTC (top) and Bitcoin-Alpha (bottom).}
    \label{fig:sensitivity}
\end{figure*}

\subsection{Parameter Sensitivity}
We evaluate how the interaction--endorsement mixing parameter $\alpha$ affects the full two-layer setting on Bitcoin-OTC and Bitcoin-Alpha. The ground-truth high/low reputation labels are held fixed throughout the sweep, so the curves measure sensitivity of the scoring methods rather than sensitivity of a moving evaluation target. For RepuLink, $\alpha$ directly changes the propagation operator, $W=\alpha T^{\top}+(1-\alpha)E^{\top}$. For each baseline, $\alpha$ is applied in the same two-layer spirit by combining its interaction-layer score with the normalized endorsement in-degree, i.e., $\alpha s_{\mathrm{base}}+(1-\alpha)s_{\mathrm{endorse}}$. Figure~\ref{fig:sensitivity} reports all four evaluation metrics for the six methods used in the main two-layer comparison.

Across both datasets, RepuLink benefits consistently from assigning more weight to the interaction layer while still retaining endorsement information. On Bitcoin-OTC, its AUC increases from $0.741$ at $\alpha=0.1$ to $0.814$ at $\alpha=0.9$, while Precision@100 rises from $0.70$ to $0.76$, Kendall's $\tau$ from $0.369$ to $0.481$, and Spearman's $\rho$ from $0.517$ to $0.652$. The same trend appears on Bitcoin-Alpha: AUC improves from $0.792$ to $0.852$, Precision@100 from $0.73$ to $0.78$, Kendall's $\tau$ from $0.377$ to $0.438$, and Spearman's $\rho$ from $0.537$ to $0.608$. The gains are monotonic or nearly monotonic across all four metrics, indicating that RepuLink is not relying on a narrow, brittle choice of $\alpha$.

The baseline curves are more method- and dataset-dependent. PageRank and EigenTrust change only mildly, AbsoluteTrust is almost flat, and ShapleyTrust slightly degrades as $\alpha$ increases. PowerTrust improves on Bitcoin-OTC but drops sharply on Bitcoin-Alpha, with AUC decreasing from $0.767$ to $0.696$. These mixed trends suggest that simple post-hoc interpolation between an interaction-only trust score and endorsement degree does not provide a stable parameter response across datasets. In contrast, RepuLink achieves the best peak value on every metric for both datasets: compared with the strongest baseline peak, its AUC is higher by about $7.8$ points on Bitcoin-OTC and $8.5$ points on Bitcoin-Alpha, and its rank correlations show similarly clear margins. Precision@100 has small low-$\alpha$ ties or near-ties, but RepuLink becomes strongest in the higher-$\alpha$ regime where the overall ranking metrics are also maximized. Overall, the sensitivity study supports the claim that RepuLink's advantage comes from jointly propagating information across the interaction and endorsement layers, rather than from selecting a favorable scalar mixing weight.

\subsection{Ablation Study}
\label{sec:ablation_study}
To evaluate the individual and combined contributions of the Backward Endorsement Penalty Propagation (BEPP) and Backward Endorsement Reward Propagation (BERP) mechanisms, we conducted an ablation study on the two datasets. This involved comparing the node rankings and distributions generated by four distinct configurations: \textbf{(i)} \textit{Forward Only (Baseline):} Standard RepuLink using only forward propagation based on Eq.~\eqref{eq:reputation}. \textbf{(ii)} \textit{With BEPP:} Baseline model augmented solely with the BEPP mechanism ($\tilde{R} = \hat{R} - \pi$). \textbf{(iii)} \textit{With BERP:} Baseline model augmented solely with the BERP mechanism ($\tilde{R} = \hat{R} + \rho$). \textbf{(iv)} \textit{Full Model:} Baseline model augmented with both BEPP and BERP mechanisms ($\tilde{R} = \hat{R} - \pi + \rho$).

\begin{figure*}[t]
    \centering
    \begin{subfigure}{0.23\linewidth}
        \includegraphics[width=\linewidth]{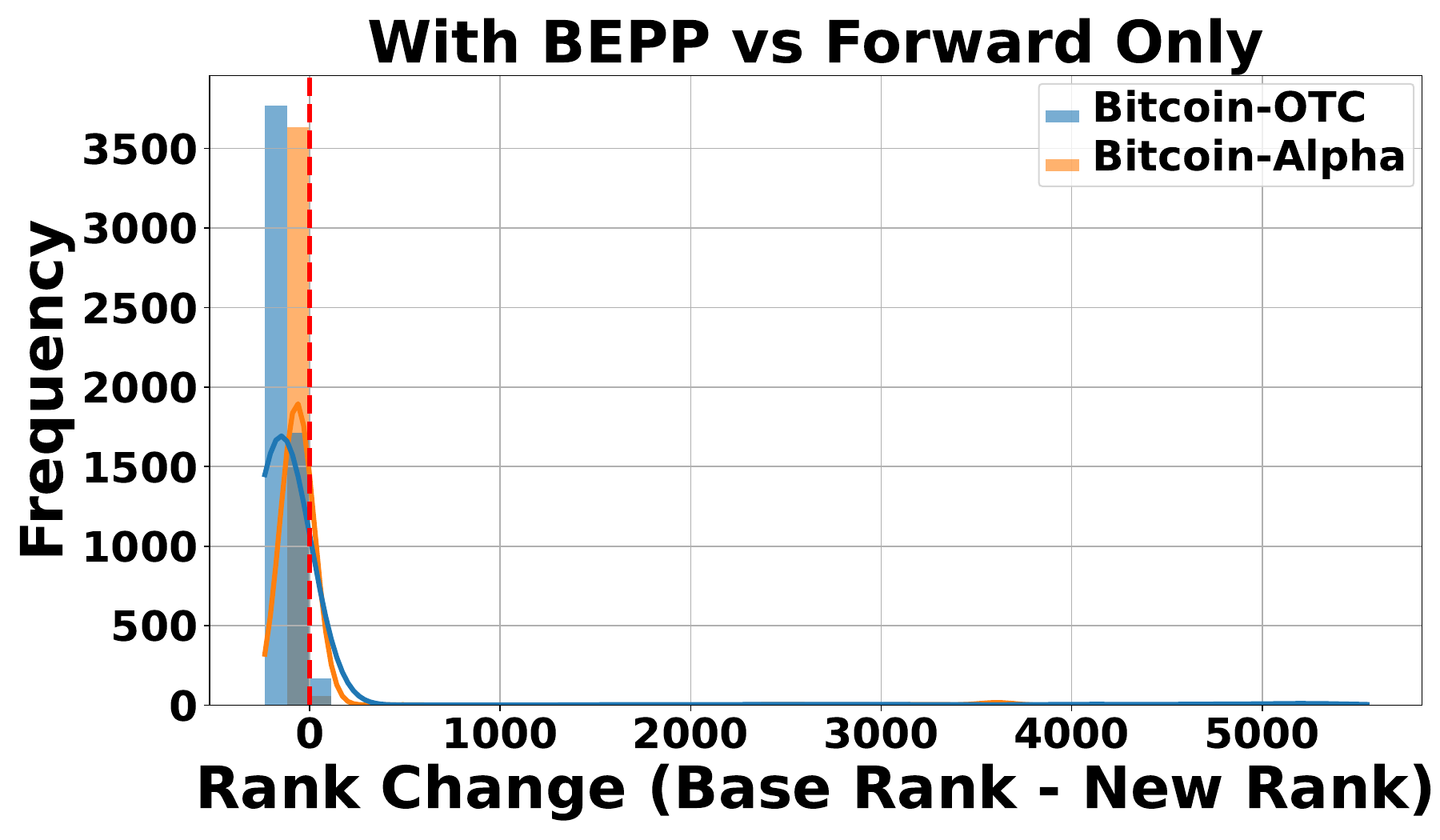}
    \end{subfigure}
    \hfill
    \begin{subfigure}{0.23\linewidth}
        \includegraphics[width=\linewidth]{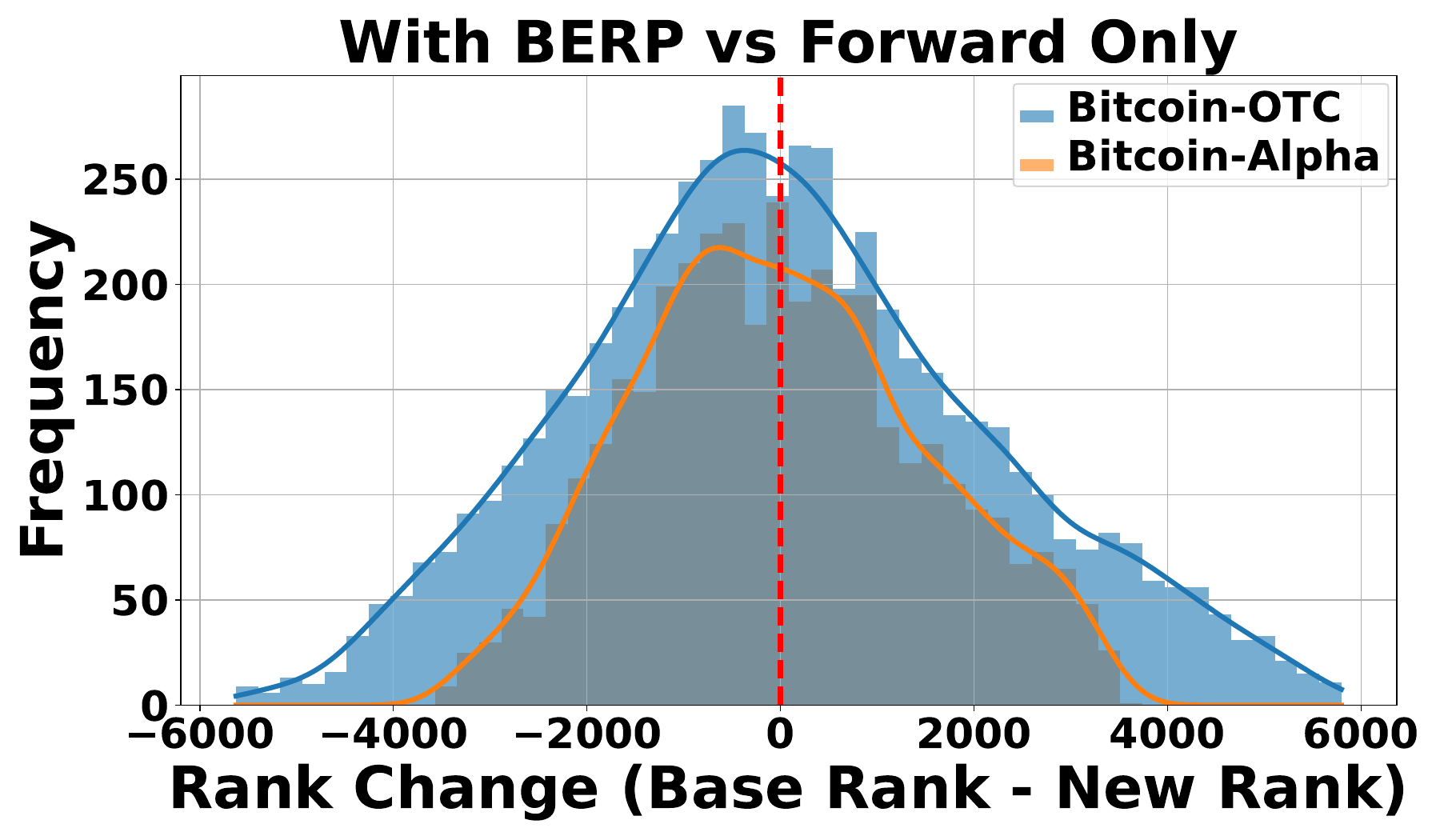}
    \end{subfigure}
    \hfill
    \begin{subfigure}{0.23\linewidth}
        \includegraphics[width=\linewidth]{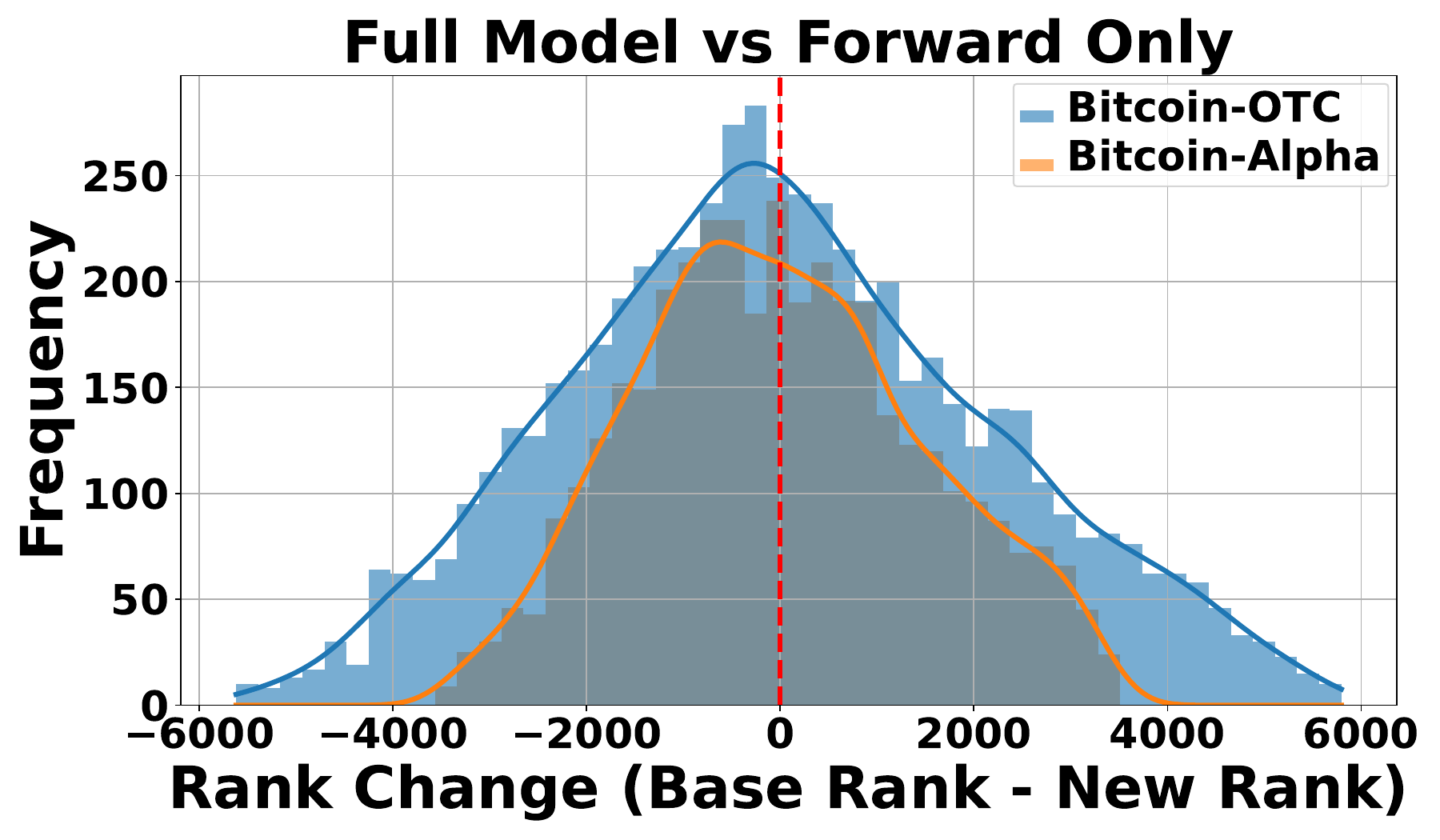}
    \end{subfigure}
    \caption{Rank Change Distribution}
    \label{fig:change_distribution}
\end{figure*}

\begin{figure*}[t]
    \centering
    \begin{subfigure}{0.18\linewidth}
        \includegraphics[width=\linewidth]{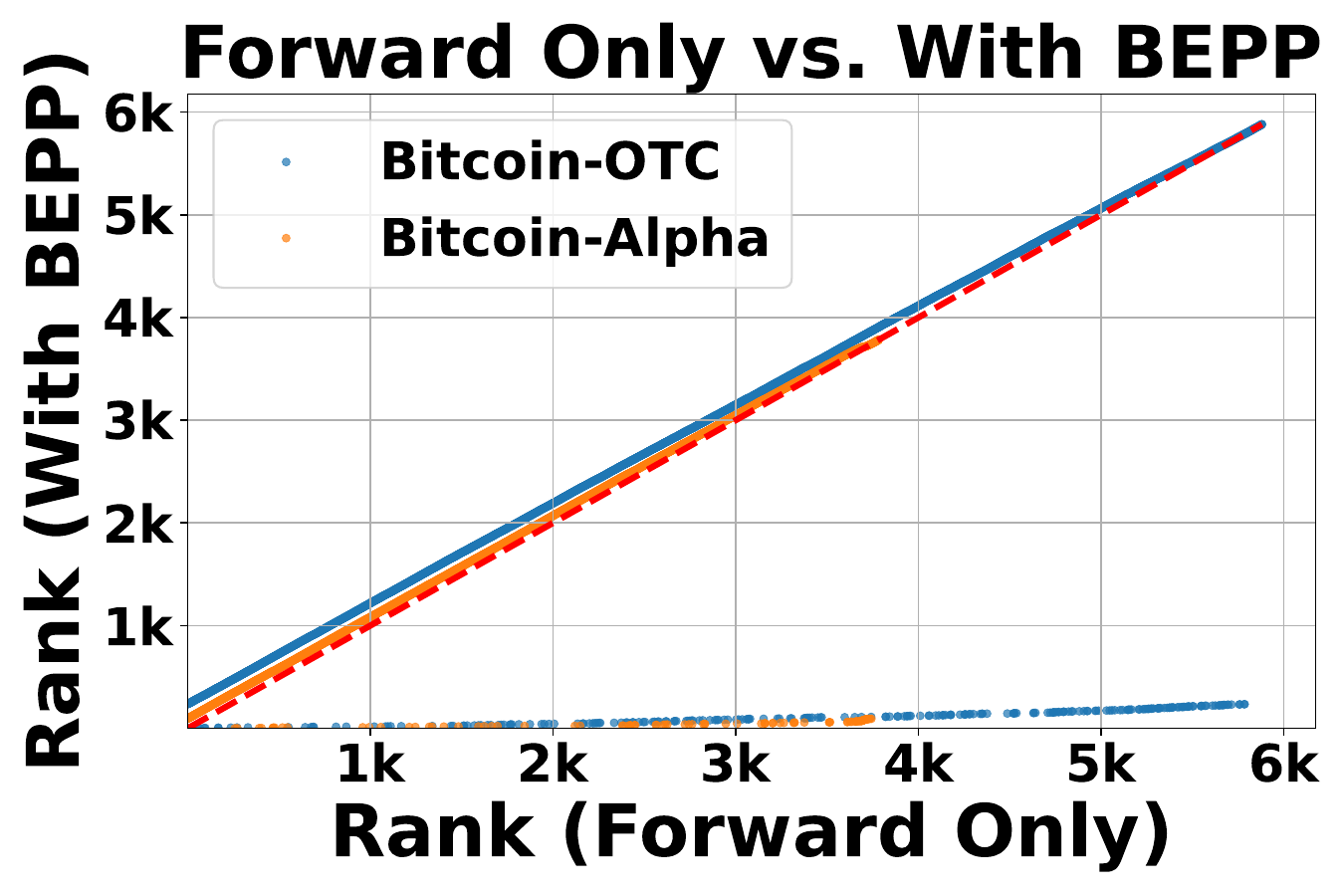}
    \end{subfigure}
    \hfill
    \begin{subfigure}{0.18\linewidth}
        \includegraphics[width=\linewidth]{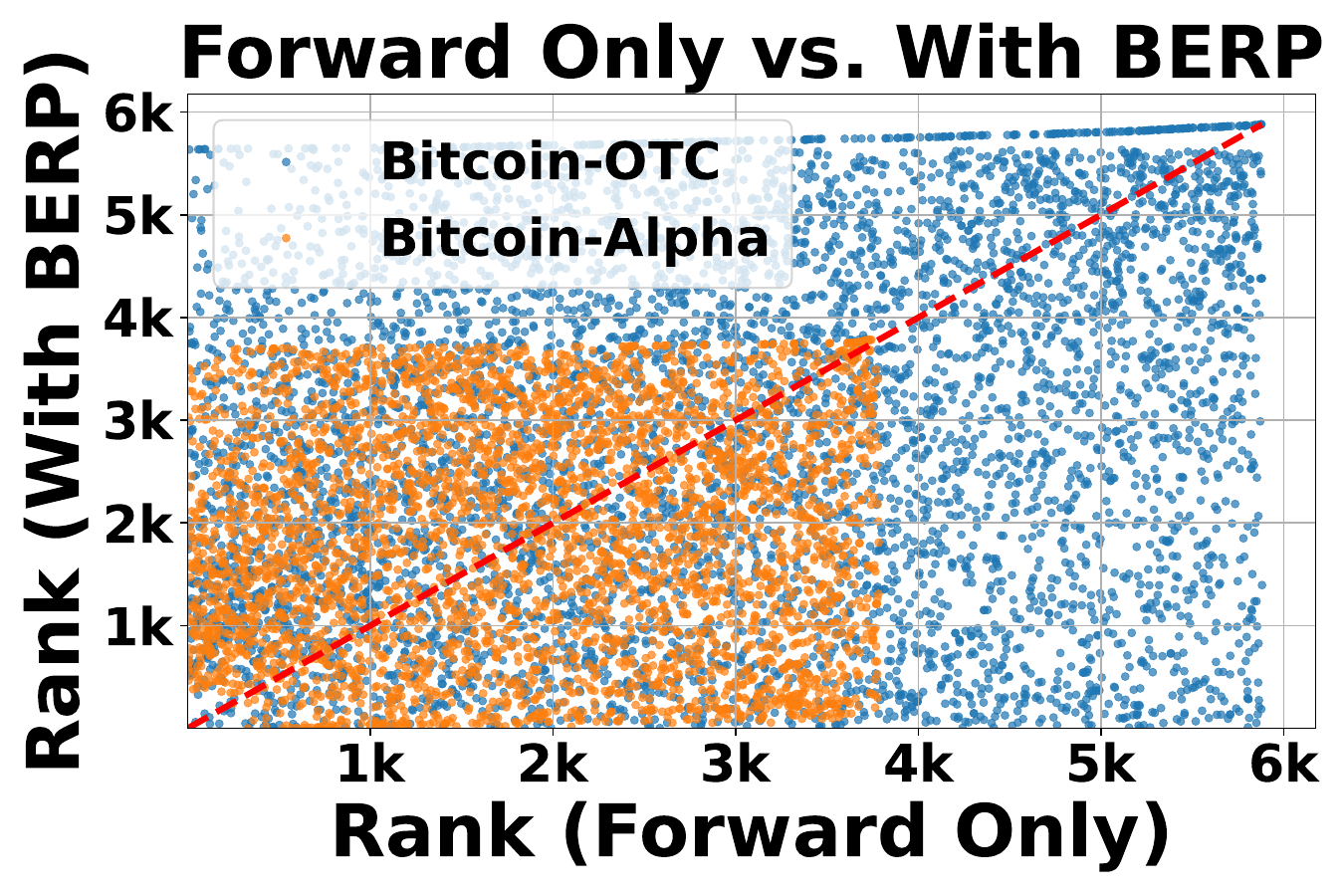}
    \end{subfigure}
    \hfill
    \begin{subfigure}{0.18\linewidth}
        \includegraphics[width=\linewidth]{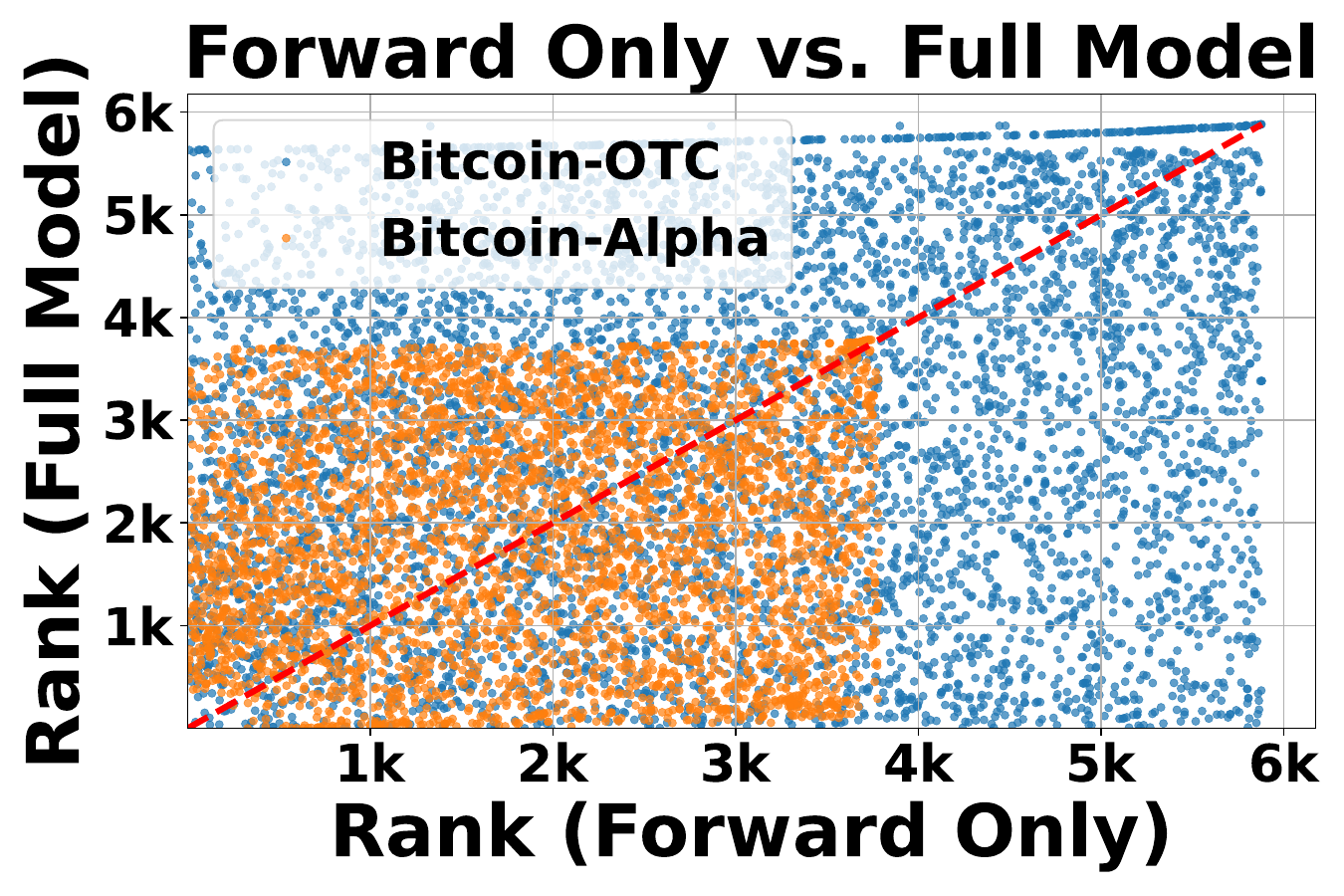}
    \end{subfigure}
    \hfill
    \begin{subfigure}{0.18\linewidth}
        \includegraphics[width=\linewidth]{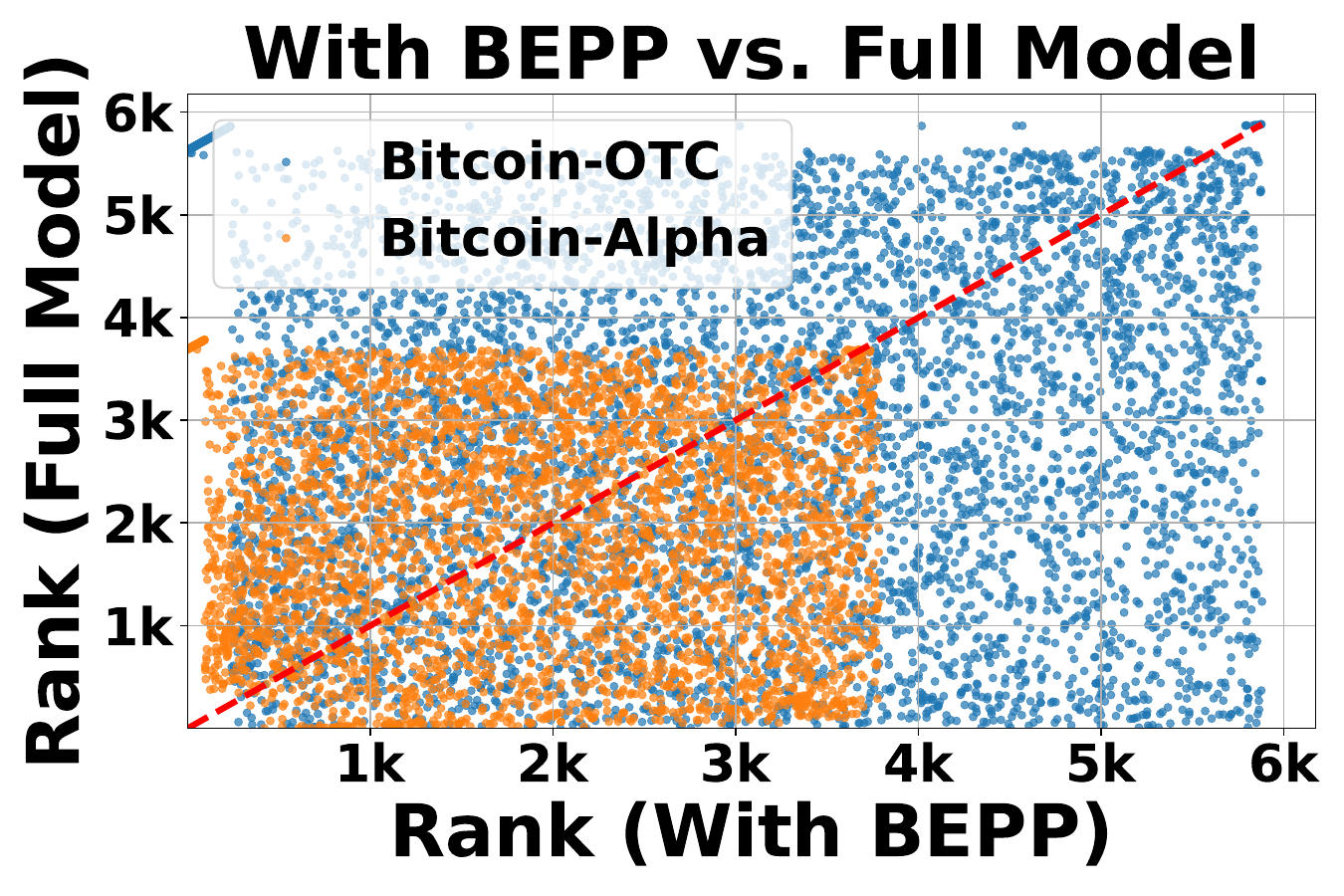}
    \end{subfigure}
    \hfill
    \begin{subfigure}{0.18\linewidth}
        \includegraphics[width=\linewidth]{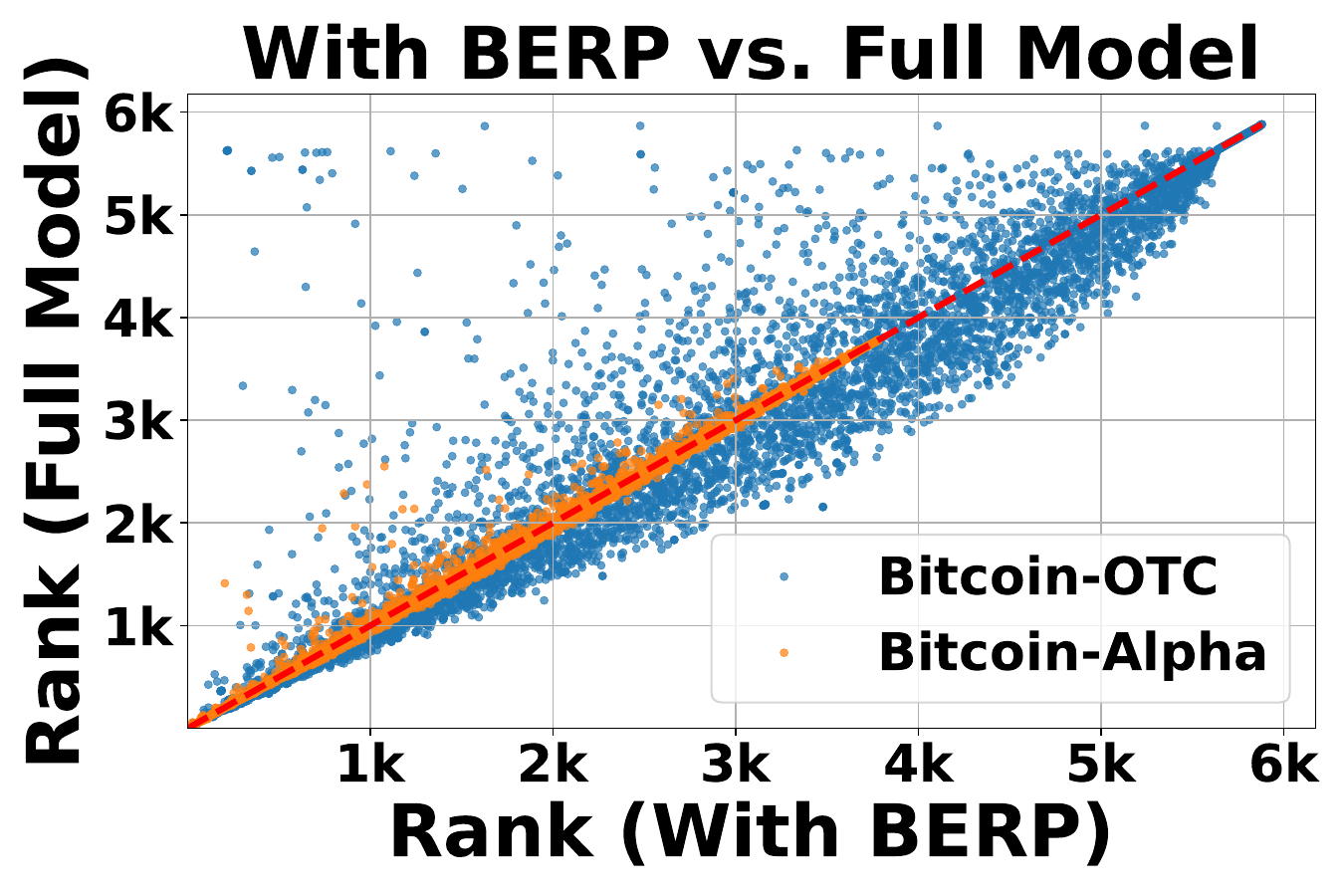}
    \end{subfigure}
    \caption{Rank Comparison}
    \label{fig:scatter}
\end{figure*}

\noindent
\textbf{Impact of Penalty Propagation (BEPP):} 
Comparing the \textit{With BEPP} model to \textit{Forward Only} reveals the specific effect of penalizing endorsers of poorly-rated nodes. The rank change histogram ($1^{st}$ figure in Fig.~\ref{fig:change_distribution}) exhibits a highly skewed distribution, with a sharp peak near zero and a long tail towards positive rank changes. This is supported by the rank comparison scatter plot ($1^{st}$ figure in Fig. \ref{fig:scatter}), where the vast majority of nodes cluster tightly along the $y=x$ line, signifying minimal rank change. However, a distinct subset of nodes shows a dramatic improvement in rank (points far below the $y=x$ line of the $1^{st}$ figure in Fig.~\ref{fig:scatter}). This indicates that \textit{BEPP} primarily acts by boosting the rank of a small number of nodes, those unaffected by or benefiting from the penalization of endorsers associated with negative feedback paths, while leaving the majority of node rankings largely unchanged or slightly decreased. The average rank change remains zero, highlighting this as a reputation redistribution rather than a global shift.

\noindent
\textbf{Impact of Reward Propagation (BERP):} 
Isolating the reward mechanism by comparing the \textit{With BERP} model to the \textit{Forward Only} shows a markedly different pattern. The rank change histogram ($2^{nd}$ of Fig.~\ref{fig:change_distribution}) displays a near-symmetrical, bell-shaped distribution centered around zero. Similarly, the rank comparison scatter plot ($2^{nd}$ of Fig.~\ref{fig:scatter}) shows points widely dispersed on both sides of the $y=x$ line, resembling the pattern observed in the \textit{Full Model} comparison ($3^{rd}$ of Fig.~\ref{fig:scatter}). This suggests that BERP induces significant, bidirectional rank changes across a broad range of nodes. Rewarding endorsers of well-rated nodes leads to rank improvements for some nodes, while the relative ranks of others decrease as part of the overall reputation redistribution. The average rank change is again zero, but the mechanism clearly drives substantial reordering throughout the network, indicating BERP is a major force in reshaping the rankings based on positive endorsement chains.

\noindent
\textbf{Combined Impact (Full Model):} 
The comparison between the \textit{Full Model} and the \textit{Forward Only} ($3^{rd}$ of Fig.~\ref{fig:change_distribution} and $3^{rd}$ of Fig.~\ref{fig:scatter}) confirms that the complete backward propagation mechanism leads to extensive rank reshuffling. The rank change distribution is roughly symmetrical and centered at zero, similar to the \textit{With BERP} case, but encompasses the effects of both BEPP and BERP. Nodes can experience rank increases or decreases depending on the net effect of the penalty and reward
propagated back from their endorsees.

\noindent
\textbf{Relative Contributions:} 
Comparing the \textit{With BEPP} and \textit{With BERP} models to the \textit{Full Model} further clarifies their roles. The scatter plot comparing \textit{With BEPP} vs. \textit{Full Model} ($4^{th}$ of Fig.~\ref{fig:scatter}) shows considerable deviation from the $y=x$ line, indicating that adding the reward mechanism significantly alters the ranks established by penalty alone. Conversely, the scatter plot comparing \textit{With BERP} vs. \textit{Full Model} ($5^{th}$ of Fig.~\ref{fig:scatter}) shows points clustered much more closely to the $y=x$ line. This suggests that while the penalty mechanism introduces important adjustments (primarily penalizing nodes linked to negative feedback, thereby causing relative rank drops compared to the \textit{With BERP} scenario), the overall ranking structure is more strongly influenced by the reward propagation. BERP appears to be the primary driver of large-scale reordering, while BEPP acts as a crucial corrective measure.

\section{Related Work} 
\label{sec:related-works}
Evaluating the importance or trustworthiness of nodes within a network is a challenge addressed by numerous ranking and trust aggregation algorithms over the past decades~\cite{vieira2007efficient,cheng2011virtual,ji2011ranking,jiang2016understanding,jin2011axiomatic,rafailidis2017learning}. 
The concept of leveraging network structure for reputation assessment has been adapted and extended extensively, particularly for trust and reputation management in peer-to-peer (P2P) systems~\cite{kamvar2003eigentrust,zhou2007powertrust,parreira2008juxtaposed} and other decentralized networks~\cite{parreira2006efficient}.

Several trust models have built upon or diverged from PageRank~\cite{page1999pagerank} to suit specific application needs. EigenTrust~\cite{kamvar2003eigentrust} computes global trust scores as the principal left eigenvector of the normalized local trust matrix, often using pre-trusted peers to enhance robustness against malicious collectives. PeerTrust~\cite{xiong2004peertrust} introduces a multi-faceted trust metric, considering factors beyond simple interactions, such as transaction feedback quantity and credibility, context, and community factors. PowerTrust~\cite{zhou2007powertrust} aims to improve efficiency and accuracy over EigenTrust~\cite{kamvar2003eigentrust} by identifying highly reputable power nodes leveraging power-law distributions observed in trust networks and using a specialized random walk strategy. Other notable approaches include PageTrust~\cite{kerchove2008pagetrust}, GossipTrust~\cite{zhou2008gossiptrust}, FAST-PPR~\cite{lofgren2014fast}, GeTrust~\cite{meng2016getrust}, and GroupTrust~\cite{fan2016grouptrust}, each offering variations in how trust evidence is gathered, aggregated, or weighted.
AbsoluteTrust \cite{awasthi2020absolutetrust} shifts the focus from relative rankings to providing absolute trust scores, facilitating direct classification of peers without normalization. ShapleyTrust~\cite{bandhana2024trust} employs cooperative game theory, using the Shapley value~\cite{chen2023algorithms} to attribute global trust based on a peer's contribution to the trust within potential coalitions.
In our experiments, we focused on comparing against the most recent together with the ``classic'' trust models.

Despite these advancements, existing models do not use domain-specific knowledge and hence face the cold-start problem~\cite{victor2008whom,guo2014merging}. Furthermore, they mainly employ forward propagation and do not support distrust, or propagating negative feedback backward through the trust edges. Peers providing misleading endorsements may not face significant repercussions, potentially compromising the integrity of the reputation scores. The importance of penalizing unreliable endorsements has been noted previously \cite{metaxas2009enhancing,kerchove2008pagetrust,victor2013enhancing,yin2025didtrust}, yet integrated accountability mechanisms remain underexplored in mainstream models. In RepuLink, we addressed these limitations.

\section{Conclusion}
\label{sec:conclusion}
In this paper, we introduce RepuLink, a novel two-layer reputation model that tackles the limitations of existing trust models regarding endorser accountability. By incorporating innovative Backward Endorsement Penalty Propagation (BEPP) and Backward Endorsement Reward Propagation (BERP) mechanisms, our model recursively adjusts node reputation based on interaction outcomes. We provide formal guarantees for the convergence of our algorithms and perform comprehensive experiments on real-world datasets. RepuLink outperforms state-of-the-art baselines across various metrics, validating that incorporating backward propagation is crucial for achieving more accurate, accountable, and robust reputation assessment. In future work, we will focus on optimizing the algorithms for dynamic and extremely large-scale networks, alongside analyzing their resistance to and performance against various security threats and malicious usage~\cite{gunes2019identifying,fan2020decentralized}.


\begin{acks}
This work was partially funded by the Horizon Europe projects DATAPACT (No. 101189771) and RAISE Suite (No. 101188337).
\end{acks}

\clearpage
\bibliographystyle{ACM-Reference-Format}
\bibliography{sample-base}

\clearpage
\appendix
\section{Proof of Convergence}\label{sec:proof-of-convergence}
We formally analyze the convergence properties of \textit{RepuLink} by examining its propagation components.
Throughout this appendix, the analysis follows the idealized normalization used for proofs: the numerical stabilizer \(c\) in the main text is set to zero, and every normalized row is assumed to have positive mass.
Under this assumption, the normalized trustworthiness matrix \(T\) and endorsement matrix \(E\) are row-stochastic.
The convergence results below should therefore be read as statements about the propagation operators in a single time slot with fixed \(T\) and \(E\).

\noindent
\textbf{Convergence of Forward Propagation.}
We first provide the convergence proof of the forward propagation in Eq.~\eqref{eq:reputation}.

\begin{definition}[Simplex]
Let \(\Delta_N := \{R \in \mathbb{R}^N \mid R_i \geq 0,\; \sum_{i=1}^N R_i = 1\}\) denote the probability simplex.
It represents all valid reputation vectors, each with non-negative entries summing to one.
\end{definition}

\begin{definition}[Forward Reputation Operator]
Let \(T,E \in \mathbb{R}^{N \times N}\) be row-stochastic matrices and let \(\alpha \in (0,1)\).
Define \(W := \alpha T^\top + (1-\alpha)E^\top\).
This operator combines interaction-based trust and endorsement-based trust.
\end{definition}

\begin{lemma}[Mapping to the Simplex]
If \(R \in \Delta_N\), then \(WR \in \Delta_N\).
\end{lemma}

\begin{proof}
Since \(T\) and \(E\) are row-stochastic, \(T^\top\) and \(E^\top\) are column-stochastic.
Thus \(W\) is column-stochastic, as \(\sum_{i=1}^N W_{ij}=\alpha \sum_{i=1}^N T^\top_{ij}+(1-\alpha)\sum_{i=1}^N E^\top_{ij}=1\).
The entries of \(W\) are non-negative, so multiplying a vector in \(\Delta_N\) by \(W\) preserves non-negativity and total mass.
Hence \(WR \in \Delta_N\).
\end{proof}

\begin{lemma}[Linearity and Continuity]
The operator \(R \mapsto WR\) is linear and continuous on \(\mathbb{R}^N\).
\end{lemma}

\begin{proof}
The matrix \(W\) is a convex combination of two fixed linear operators, \(T^\top\) and \(E^\top\).
Therefore \(R \mapsto WR\) is linear.
All finite-dimensional linear maps are continuous.
\end{proof}

\begin{theorem}[Single-Slot Convergence of Forward Propagation]
Fix a time slot and assume \(T\) and \(E\) are fixed row-stochastic matrices.
Let \(W=\alpha T^\top+(1-\alpha)E^\top\).
If \(W\) is primitive, then the iteration \(R_{k+1}=WR_k\), \(R_0\in\Delta_N\), converges to a unique vector \(R^*\in\Delta_N\) satisfying \(WR^*=R^*\).
\end{theorem}

\begin{proof}
By the previous lemma, \(W\) is column-stochastic and maps \(\Delta_N\) into itself.
Since \(W\) is primitive, the Perron-Frobenius theorem implies that the eigenvalue \(1\) is simple and dominant, and that the corresponding eigenvector can be chosen strictly positive.
After normalizing this eigenvector to have \(\ell_1\)-norm one, we obtain a unique \(R^*\in\Delta_N\).
Moreover, for every \(R_0\in\Delta_N\), the power iteration \(R_{k+1}=WR_k\) converges to \(R^*\).
\end{proof}

\begin{remark}[Relation to the Full RepuLink Update]
The result above proves convergence of the forward propagation operator within a single time slot, where \(T\) and \(E\) are fixed.
For the full dynamic RepuLink process, \(T\), \(E\), and the backward penalty/reward signals may change across time slots.
Thus, unconditional global convergence requires additional stability assumptions.
For example, if the time-varying matrices and propagated signals eventually stabilize, and the limiting forward operator is primitive, then the single-slot convergence result extends to the limiting RepuLink update by continuity of the update components.
\end{remark}

\noindent
\textbf{Convergence of Backward Propagation.}
We next analyze the penalty and reward propagation mechanisms through a unified recursive signal model.

\begin{definition}[Backward Signal Vectors]
Let \(s_\pi := \mathbf{1}-g(\mathcal{N})\) and \(s_\rho := r(\mathcal{P})-\mathbf{1}\).
Here \(s_\pi\) is the penalty signal propagated in Eq.~\eqref{eq:endorser penalty}, and \(s_\rho\) is the reward signal propagated in Eq.~\eqref{eq:endorser rewards}.
For finite feedback scores, \(s_\pi\) and \(s_\rho\) are bounded and non-negative.
\end{definition}

\begin{definition}[Recursive Signal Propagation]
Let \(E\in\mathbb{R}^{N\times N}\) be the row-stochastic endorsement matrix and let \(s\in\mathbb{R}^N\) be a bounded signal vector, such as \(s_\pi\) or \(s_\rho\).
For \(\gamma\in(0,1)\), define \(\phi(s) := \sum_{k=1}^{\infty} \gamma^k E^k s\).
\end{definition}

\begin{lemma}[Bounded Recursive Signal Terms]\label{lemma:signal-terms}
Let \(E\) be row-stochastic and suppose \(\|s\|_\infty \leq M\).
Then, for every \(k\geq 1\), \(\|\gamma^k E^k s\|_\infty \leq \gamma^k M\).
\end{lemma}

\begin{proof}
For a row-stochastic matrix, \(\|E\|_\infty=1\).
By submultiplicativity of the induced norm, \(\|E^k s\|_\infty \leq \|E^k\|_\infty\|s\|_\infty \leq \|E\|_\infty^k \|s\|_\infty \leq M\).
Multiplying by \(\gamma^k\) gives the result.
\end{proof}

\begin{lemma}[Absolute Convergence of Recursive Signal Series]\label{lemma:signal-series}
The series \(\sum_{k=1}^{\infty}\gamma^k E^k s\) converges absolutely in the \(\ell_\infty\) norm.
\end{lemma}

\begin{proof}
By Lemma~\ref{lemma:signal-terms}, \(\sum_{k=1}^{\infty}\|\gamma^k E^k s\|_\infty \leq \sum_{k=1}^{\infty}\gamma^k M = \gamma M/(1-\gamma)<\infty\).
Hence the recursive signal series converges absolutely.
\end{proof}

\begin{theorem}[Convergence of Recursive Signal Propagation]
Let \(E\) be row-stochastic, let \(\gamma\in(0,1)\), and let \(s\) be bounded.
Then \(\phi(s)=\sum_{k=1}^{\infty}\gamma^k E^k s\) converges and admits the closed form \(\phi(s)=\gamma E{(I-\gamma E)}^{-1}s\).
\end{theorem}

\begin{proof}
The convergence follows from Lemma~\ref{lemma:signal-series}.
Since \(\|E\|_\infty=1\), we have \(\|\gamma E\|_\infty=\gamma<1\).
Thus the Neumann series is valid, \(\sum_{k=0}^{\infty}{(\gamma E)}^k = {(I-\gamma E)}^{-1}\).

\noindent
Therefore, \(\phi(s)=\sum_{k=1}^{\infty}\gamma^k E^k s=\gamma E\sum_{k=0}^{\infty}{(\gamma E)}^k s=\gamma E{(I-\gamma E)}^{-1}s\).
\end{proof}

\begin{remark}
Taking \(s=s_\pi=\mathbf{1}-g(\mathcal{N})\) gives the infinite-depth counterpart of the backward endorsement penalty vector \(\pi\), and taking \(s=s_\rho=r(\mathcal{P})-\mathbf{1}\) gives the corresponding reward vector \(\rho\).
The finite-\(K\) versions in the main text are truncated approximations of these limits.
\end{remark}

\noindent
\textbf{Normalization and Fixed Points.}
After backward adjustment and the subsequent reputation propagation step in Alg.~\ref{alg:iterative_reputation_update}, RepuLink normalizes the corrected reputation vector.
In the proof setting \(c=0\), the normalization is well-defined whenever the positive part of the corrected vector has nonzero mass.

\begin{definition}[Positive Simplex Projection]
For any \(x\in\mathbb{R}^N\) with \(\|\max(x,0)\|_1>0\), define \(\Pi_+(x):=\max(x,0)/\|\max(x,0)\|_1\).
\end{definition}

\begin{lemma}[Continuity of Normalization]\label{lemma:normalization-continuity-adjusted}
The mapping \(\Pi_+\) is continuous on the domain \(\Omega := \{x\in\mathbb{R}^N \mid \|\max(x,0)\|_1>0\}\).
Moreover, \(\Pi_+(x)\in\Delta_N\) for every \(x\in\Omega\).
\end{lemma}

\begin{proof}
The element-wise map \(x\mapsto \max(x,0)\) is continuous, and the denominator \(\|\max(x,0)\|_1\) is continuous and strictly positive on \(\Omega\).
Thus \(\Pi_+\) is continuous on \(\Omega\).
By construction, \(\Pi_+(x)\) is non-negative and has \(\ell_1\)-norm one, so \(\Pi_+(x)\in\Delta_N\).
\end{proof}

\begin{theorem}[Existence of a Fixed Point for a Fixed-Slot Normalized Update]\label{thm:normalization-fixedpoint-adjusted}
Fix a time slot and treat \(W\), \(\pi\), and \(\rho\) as fixed.
Define the normalized update map \(\mathcal{U}(R):=\Pi_+\bigl(W(R-\pi+\rho)\bigr)\), \(R\in\Delta_N\), which follows the update order in Alg.~\ref{alg:iterative_reputation_update}.
Assume \(W(R-\pi+\rho)\in\Omega\) for every \(R\in\Delta_N\).
Then \(\mathcal{U}\) has at least one fixed point \(R^*\in\Delta_N\).
\end{theorem}

\begin{proof}
The map \(R\mapsto W(R-\pi+\rho)\) is continuous, and by assumption its image lies in \(\Omega\).
By Lemma~\ref{lemma:normalization-continuity-adjusted}, \(\Pi_+\) is continuous on this image.
Therefore \(\mathcal{U}\) is a continuous self-map from the compact convex set \(\Delta_N\) to itself.
Brouwer's fixed-point theorem guarantees the existence of at least one \(R^*\in\Delta_N\) such that \(\mathcal{U}(R^*)=R^*\).
\end{proof}

\begin{remark}
The fixed-point result establishes existence for the fixed-slot normalized update.
It does not by itself imply uniqueness or convergence of repeated normalized updates; such claims require stronger contraction or monotonicity assumptions.
\end{remark}

\section{Convergence Speed Analysis}\label{sec:convergence-speed-analysis}
We analyze the convergence rates of the three major components of the reputation update procedure: forward propagation, backward propagation, and normalization.

\noindent
\textbf{Forward Propagation.}
Let \(W:=\alpha T^\top+(1-\alpha)E^\top\).
Assume \(W\) is primitive and column-stochastic.

\begin{definition}[Second Largest Eigenvalue Modulus]
Let \(\mu_1=1,\mu_2,\ldots,\mu_N\) be the eigenvalues of \(W\), ordered so that \(|\mu_1|=1\), and define \(\sigma := \max_{i\geq 2}|\mu_i| < 1\).
The value \(\sigma\) is the second largest eigenvalue modulus (SLEM).
\end{definition}

\begin{theorem}[Convergence Rate of Forward Propagation]
Let \(R_k\) be the forward propagation iterate in a fixed time slot.
For any \(\eta\) satisfying \(\sigma<\eta<1\), there exists a constant \(\kappa_\eta>0\) such that, for all \(R_0\in\Delta_N\), \(\|R_k-R^*\|_1 \leq \kappa_\eta \eta^k\).
If \(W\) is diagonalizable, the same bound can be stated with \(\eta=\sigma\) for a suitable constant.
\end{theorem}

\begin{proof}
Since \(W\) is primitive and column-stochastic, the Perron eigenvalue \(1\) is simple and all other eigenvalues have modulus strictly less than one.
The Jordan decomposition of \(W\) implies that the non-stationary component of \(W^k R_0\) decays geometrically at any rate \(\eta\) larger than the largest non-Perron eigenvalue modulus.
Because all norms are equivalent in finite dimensions, this yields the stated \(\ell_1\) bound.
If \(W\) is diagonalizable, no polynomial Jordan factor is present, and the decay rate can be written directly in terms of \(\sigma\).
\end{proof}

\begin{definition}[Mixing Time]
For \(\varepsilon>0\), define the fixed-slot \(\varepsilon\)-mixing time as \(\tau_{\mathrm{mix}}(\varepsilon):=\min\{k\in\mathbb{N}:\sup_{R_0\in\Delta_N}\|R_k-R^*\|_1 \leq \varepsilon\}\).
\end{definition}

\begin{corollary}
Under the assumptions of the previous theorem, \(\tau_{\mathrm{mix}}(\varepsilon)\leq \log(\kappa_\eta/\varepsilon)/(-\log \eta)\leq \log(\kappa_\eta/\varepsilon)/(1-\eta)\).
\end{corollary}

\noindent
\textbf{Backward Propagation.}
Let \(s\) denote either \(s_\pi=\mathbf{1}-g(\mathcal{N})\) or \(s_\rho=r(\mathcal{P})-\mathbf{1}\), and assume \(\|s\|_\infty\leq M\).
The \(K\)-term approximation is \(\phi^{(K)}(s) := \sum_{k=1}^{K}\gamma^k E^k s\).
The residual error is \(\mathcal{D}^{(K)}:=\phi(s)-\phi^{(K)}(s)=\sum_{k=K+1}^{\infty}\gamma^k E^k s\).
Using Lemma~\ref{lemma:signal-terms}, \(\|\mathcal{D}^{(K)}\|_\infty \leq \gamma^{K+1}\|s\|_\infty/(1-\gamma)\).
Thus, to ensure \(\|\mathcal{D}^{(K)}\|_\infty\leq\varepsilon\), it suffices to choose \(K \geq \left\lceil \log(\varepsilon(1-\gamma)/M)/\log \gamma -1 \right\rceil\), with \(K\geq 0\) and \(\log\gamma<0\).
This shows exponential convergence of backward signal diffusion due to the geometric decay of \(\gamma^k\).

\noindent
\textbf{Normalization Effects.}
The normalization operator is not globally Lipschitz when \(c=0\), because the denominator can approach zero.
However, it is Lipschitz on any subset where the positive mass is bounded away from zero.

\begin{lemma}[Local Lipschitz Continuity of Normalization]
Let \(\Omega_m := \{x\in\mathbb{R}^N \mid \|\max(x,0)\|_1 \geq m\}\) for some \(m>0\).
For any \(x,y\in\Omega_m\), \(\|\Pi_+(x)-\Pi_+(y)\|_1 \leq (2/m)\|x-y\|_1\).
\end{lemma}

\begin{proof}
Let \(x_+=\max(x,0)\), \(y_+=\max(y,0)\), \(a=\|x_+\|_1\), and \(b=\|y_+\|_1\).
Since \(a,b\geq m\), \(\|x_+/a-y_+/b\|_1 \leq \|x_+-y_+\|_1/a+\|y_+\|_1|1/a-1/b|\).
The rectifier is \(1\)-Lipschitz in \(\ell_1\), and \(|a-b|\leq \|x_+-y_+\|_1\).
Therefore, \(\|\Pi_+(x)-\Pi_+(y)\|_1 \leq \|x-y\|_1/m+\|x-y\|_1/m=(2/m)\|x-y\|_1\).
\end{proof}

\section{Computational Complexity Analysis}\label{sec:computational-complexity-analysis}
We analyze the computational complexity of one global RepuLink update round.
Let \(N\) denote the number of nodes, \(K\) the maximum number of backward propagation iterations, \(m_{\mathcal{F}}:=|\mathcal{F}|\) the number of sparse interaction-feedback edges, and \(m_{\mathcal{E}}:=|\mathcal{E}|\) the number of endorsement edges.
After constructing the normalized matrices, \(\operatorname{nnz}(T)\leq m_{\mathcal{F}}\) and \(\operatorname{nnz}(E)=m_{\mathcal{E}}\).

\noindent
\textbf{Trustworthiness and Feedback Signals.}
Computing the aggregated feedback vectors \(\mathcal{N}\) and \(\mathcal{P}\), and then evaluating \(g(\mathcal{N})\) and \(r(\mathcal{P})\), costs \(\mathcal{O}(m_{\mathcal{F}}+N)\), assuming feedback is stored sparsely over interaction edges.
Normalizing the trustworthiness matrix \(T\) also costs \(\mathcal{O}(m_{\mathcal{F}})\).

\noindent
\textbf{Recursive Backward Propagation.}
The propagated penalty and reward vectors are computed as \(\pi^{(K)}=\sum_{k=1}^{K}\gamma^k E^k(\mathbf{1}-g(\mathcal{N}))\) and \(\rho^{(K)}=\sum_{k=1}^{K}\gamma^k E^k(r(\mathcal{P})-\mathbf{1})\).
Each iteration requires a sparse multiplication by \(E\), costing \(\mathcal{O}(m_{\mathcal{E}})\).
Computing both \(\pi\) and \(\rho\) costs \(\mathcal{O}(K m_{\mathcal{E}})\).

\noindent
\textbf{Endorsement Update.}
Updating the nonzero entries of \(\hat{E}\) by the factor \(g(\mathcal{N}_j)r(\mathcal{P}_j)\), and then normalizing \(E\), costs \(\mathcal{O}(m_{\mathcal{E}})\).

\noindent
\textbf{Reputation Propagation and Normalization.}
Following Alg.~\ref{alg:iterative_reputation_update}, RepuLink first applies the backward adjustment \(\hat{R}^{(t+1)}=R^{(t)}-\pi+\rho\), which costs \(\mathcal{O}(N)\).
It then computes \(\tilde{R}^{(t+1)}=\alpha T^\top\hat{R}^{(t+1)}+(1-\alpha)E^\top\hat{R}^{(t+1)}\).
This step involves two sparse matrix-vector multiplications and costs \(\mathcal{O}(m_{\mathcal{F}}+m_{\mathcal{E}})\).
Normalizing the final reputation vector costs another \(\mathcal{O}(N)\).

\noindent
Therefore, the total cost of one global update round is
\(\mathcal{O}(m_{\mathcal{F}}+K m_{\mathcal{E}}+N)\).
If \(m_{\mathcal{F}}=\mathcal{O}(d_T N)\) and \(m_{\mathcal{E}}=\mathcal{O}(d_E N)\), this becomes \(\mathcal{O}((d_T+Kd_E)N)\).
When both layers have average degree \(\mathcal{O}(d)\) and \(K\) is treated as a small constant, the update is linear in \(N\) and \(d\).

\end{document}